\documentclass[sn-basic
]{sn-jnl}
\usepackage[utf8]{inputenc}

\usepackage{tikz}
\usetikzlibrary{matrix,arrows,decorations.pathreplacing,tikzmark,calc}
\tikzset{break/.style={fill=white, inner sep=1pt}}
\tikzset{inline/.style={column sep = 2.5em,inner sep = 1pt}}
\usetikzlibrary{arrows,positioning}

\usepackage{graphicx}%
\usepackage{multirow}%
\usepackage{amsmath,amssymb,amsfonts}%
\usepackage{amsthm}%
\usepackage{mathrsfs}%
\usepackage[title]{appendix}%
\usepackage{xcolor}%
\usepackage{textcomp}%
\usepackage{manyfoot}%
\usepackage{booktabs}%
\usepackage{algorithm}%
\usepackage{algorithmicx}%
\usepackage{algpseudocode}%
\usepackage{listings}%
\usepackage{lmodern}

\usepackage{amsmath,amsfonts,amsthm,amssymb,thmtools,mathtools,mathrsfs,dsfont}
\usepackage{graphics}
\usepackage{epstopdf}
\usepackage{float}
\usepackage{microtype}
\usepackage{lineno}
\usepackage{natbib}

\usepackage{array}
\newcolumntype{L}[1]{>{\raggedright\let\newline\\\arraybackslash\hspace{0pt}}m{#1}}
\newcolumntype{C}[1]{>{\centering\let\newline\\\arraybackslash\hspace{0pt}}m{#1}}
\newcolumntype{R}[1]{>{\raggedleft\let\newline\\\arraybackslash\hspace{0pt}}m{#1}}

\theoremstyle{plain}
\newtheorem{theorem}{Theorem}[section]

\newtheorem{corollary}[theorem]{Corollary}
\theoremstyle{definition}
\newtheorem{definition}[theorem]{Definition}
\newtheorem{example}[theorem]{Example}

\title{Disease Transmission on Random Graphs Using Edge-Based Percolation}

\author[1]{\fnm{S.} \sur{Zhao}}\email{20sz11@queensu.ca}

\author*[1]{\fnm{F.M.G.} \sur{Magpantay}}\email{felicia.magpantay@queensu.ca}

\affil*[1]{\orgdiv{Department of Mathematics and Statistics}, \orgname{Queen's University}, \orgaddress{\street{48 University Avenue}, \city{Kingston}, \postcode{K7L 3N6}, \state{ON}, \country{Canada}}}

\abstract{
Edge-based percolation methods can be used to analyze disease transmission on complex social networks. This allows us to include complex social heterogeneity in our models while maintaining tractability. 
Here we begin by reviewing the seminal works on this field by \cite{NewmanStrogatzWatts:2001, Newman:2002, Newman:2003}, and \cite{MillerSlimVolz:2012}. 
We present a unified discussion of the theoretical background behind these models, including a detailed derivation of some of the major results. We also demonstrate the connections between these papers and  the classical literature in random graph theory~\cite{MolloyReed:1995, MolloyReed:1998}.
Finally, we also present an accompanying R package that takes epidemic and network parameters as input and generates estimates of the epidemic trajectory and final size.
This manuscript and the R package were developed to help researchers easily understand and use network models to investigate the interaction between different community structures and disease transmission.
}


\pacs[Acknowledgements]{
We are grateful to T. Day and K. Zhang for helpful discussions. 
We acknowledge the support of the Natural Sciences and Engineering Research Council of Canada (NSERC) Discovery Grant Program and the Mathematics for Public Health network of the NSERC Emerging Infectious Diseases Modelling Initiative.}

\begin{document}

\maketitle

\section{Introduction}
\label{sec: intro}

Models of disease transmission are commonly set up as compartmental models whose dynamics are controlled by a system of ordinary differential equations (ODEs). 
The susceptible-infectious-recovered (SIR) is the basis of many such models and this often comes with a mass action (MA) assumption wherein the population is assumed to be homogeneous and fully mixed.
We refer to this model as the MA-SIR model and a diagram is shown in Figure~\ref{fig: MA-SIR}.

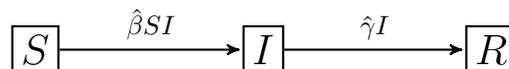
\begin{figure}[htbp]
    \begin{center}
        \begin{tikzpicture}[->, >=stealth',shorten >=1pt,auto, node distance=3cm, thick, main node/.style={rectangle,draw, font = \sffamily\Large\bfseries}]

            \node[main node] (S) {$S$};
            \node[main node] (I) [right of=S] {$I$};
            \node[main node] (R) [right of=I] {$R$};
            \draw[thick,->] (S) -- (I) node[midway,sloped,above,rotate=0] {$\hat{\beta} S I$};
	
            \draw[thick,->] (I) -- (R) node[midway,sloped,above,rotate=0] {$\hat{\gamma} I$};

        \end{tikzpicture}
    \end{center}
    \caption[The MA-SIR model]{The MA-SIR model. Here $S$, $I$ and $R$ are the susceptible, infectious and recovered/removed compartments respectively. Under the mass action (MA) assumption, the rate of flow from $S$ to $I$ is equal to their product times a constant per-infected transmission rate of $\hat{\beta}$. Infectious individuals are assumed to recover at a per capita rate of $\hat{\gamma}$.
    Figure was created using Latex Tikz Picture.}
\label{fig: MA-SIR}
\end{figure}

The assumptions of MA transmission and perfect infection-derived immunity keeps the dynamics of the MA-SIR model relatively tractable for various kinds of analysis, such as the computation of the final epidemic size.
The MA assumption is reasonable for some diseases spreading across highly connected populations.
However, for more structured populations or for cases when transmission requires specific types of contact, these assumptions can be problematic~\citep{BansalMeyers:2012}.
The MA assumption can lead to overestimates of the size of the actual outbreak, as demonstrated in Figure~\ref{fig: epidemic_curves} where we compare an MA-SIR model with an SIR model with a heterogeneous contact network. 

\begin{figure}[htbp]
    \begin{center}
        \includegraphics[width=0.6\textwidth]{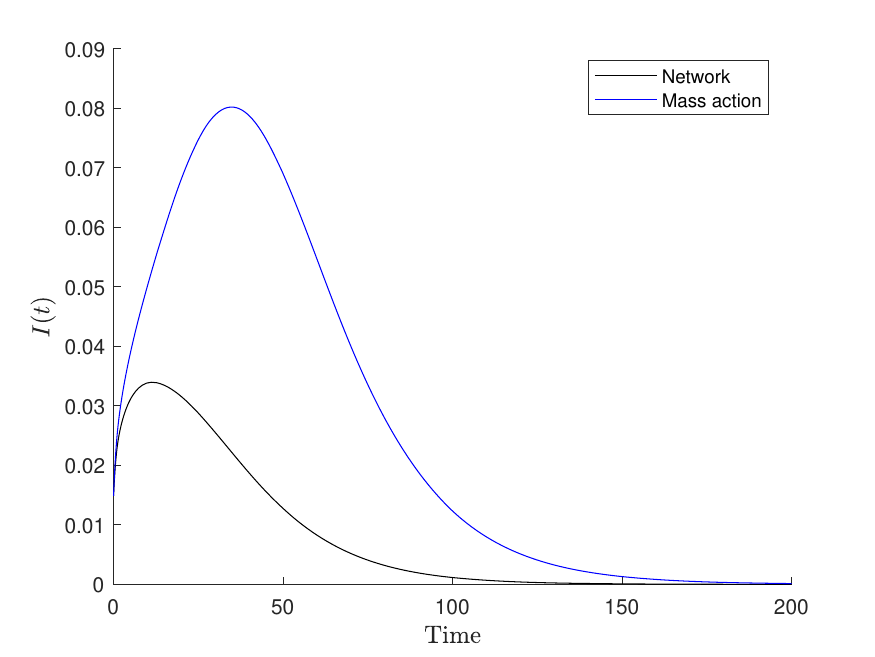}
    \end{center}
    \caption[Epidemic curves of an MA-SIR and a network-SIR models]{Epidemic curves of an MA-SIR model and a network-SIR model by \citet{MillerSlimVolz:2012} with the same basic reproduction numbers. Figure was created using Matlab.}
\label{fig: epidemic_curves}
\end{figure}

There have been many modeling studies that attempt to eliminate the MA assumption without significantly sacrificing analytical tractability.
Some approaches involve introducing social heterogeneity by adding multiple strains and risk groups for each compartment.
Others incorporate age-specific contact rates and age-structured populations (by creating compartments for different age groups or using partial differential equation models).
Here we consider the use of complex networks of connected individuals over which the disease can spread. 
Complex network approaches have been gaining in popularity recently and are often considered to be a good way to introduce social heterogeneity into compartment models with more flexibility about partnership duration, while keeping dynamic information and tractability.
  
Epidemiology is one of the main application of networks in biomathematics. 
Most of these models employ a network composed of an undirected graph whose vertices represent individuals in a host population and edges are social connections that allow infection to be transmitted between individuals.
The type, degree of vertices or weight of edges can be amended to fit different assumptions of epidemics.
\citet{Newman:2003} provides a comprehensive exposition of different techniques and models developed in recent years for disease on complex networks. 
Here we follow his approach, terminology, classification and theory for networks.
This paper is organized as follows: In Section~\ref{sec: review} we establish our notation and review the theoretical foundations of random graphs. In Section~\ref{sec: disease} we introduce percolation-based transmission over networks following~\citet{Newman:2003}.
We also clarify the connections between random graph theory and network epidemiology.
In Section~\ref{sec: dyn} we  review the dynamics of network models from \cite{MillerSlimVolz:2012} and relate some results to the previous sections.
In Section~\ref{sec:summary} we summarize our findings and discuss future directions.

\section{Theory of random graphs}
\label{sec: review}

The mathematical structures we use to describe networks are graphs.
In this section we review some definitions and results from basic graph theory and theories specific for random graphs.

\subsection{Basic definitions}
\label{sec: basic}

\begin{definition} [\cite{BondyMurty:2008}]
\label{def: Graph & Subgraph}
An \textbf{undirected graph} $G$ is an ordered pair $G=(V_G,E_G)$ consisting of a set $V_G$ of \textbf{vertices} and a set $E_G$, disjoint from $V_G$, of \textbf{undirected edges}, together with an incidence function $\psi_G$ that associates with each edge $e\in E_G$ of $G$ an unordered pair $\{a,b\}$ of vertices $a,b \in V_G$ of $G$. 
We also use the following terminology:

\begin{itemize}

    \item If for $e \in E$ and $a,b \in V$, $\psi_G(e)=\{a,b\}$, then $e$ is said to \textbf{join} (or \textbf{connect}) $a$ and $b$, the vertices $a, b$ are called the \textbf{ends} of $e$.
    We also say that the ends $a, b$ of an edge $e$ are  \textbf{incident} with $e$, and vice versa. 

    \item A \textbf{loop} is an edge with identical ends. In other words, a loop is an edge join a vertex with itself.

    \item A graph is said to have \textbf{parallel edges}  if there exist two or more edges with the same pair of ends.

    \item A graph is \textbf{simple} if it has no loops or parallel edges.

    \item A graph $F=(V_F,E_F)$ is called a \textbf{subgraph} of a graph $G=(V_G,E_G)$, denoted by $F \subseteq G$, if $V_F \subseteq V_G$ and $E_F \subseteq E_G$, and $\psi_F=\psi_G|_{E_F}$ is the restriction of $\psi_G$ to $E_F$. 
    If $F \subseteq G$, $G$ is also said to \textbf{contain} $F$.
\end{itemize}

\end{definition}

\begin{theorem}
\label{cor: Simple Subgraph}
A subgraph of a simple undirected graph is a simple undirected graph.
\end{theorem}
\begin{proof}
This follows directly from the definitions of simple, undirected graphs and subgraphs.
\end{proof}

To consider properties of graphs, we introduce the following notation given in Definition~\ref{def: size and degree}.

\begin{definition}[\cite{BondyMurty:2008}]
\label{def: size and degree}
Consider a simple, undirected graph $G=(V_G,E_G)$.
\begin{itemize}

\item The \textbf{size} of the graph $G$, often denoted by $N$, is the number of vertices in $G$, i.e. $N=|V_G|$. 

\item The \textbf{total number of edges} of $G$, often denoted by $m$, is the number of elements of set $E_G$, thus $m=|E_G|$.

\item Two distinct vertices $a,b \in V_G$ are said to be \textbf{neighbours} if they are connected by an edge $e \in E_G$. $a$ and $b$ are also said to be \textbf{neighbour} of each other. The set of all neighbours of a vertex $a$ in graph $G$ is denoted by $N_G(a)$
    
\item The \textbf{degree} of a vertex $a$ in graph $G$, denoted by $D_G(a)$, is the number of edges in $G$ that incident with $a$. 
 For a simple graph $G$, degree of $u$ is also the number of neighbours of $a$, thus $D_G(a)=|N_G(a)|$.

\item  An \textbf{isolated vertex} is a vertex of degree 0, which means no edge is incident with it and $N_G(a) = \emptyset$.
        
\end{itemize}
\end{definition}

In the network models we will consider, vertices represent individuals in the community and edges represent a disease-transmitting connection between the corresponding individuals, which allows transmission in both directions.
We note that even if two individuals are strongly socially connected, they might not have a disease-causing connection.
For example, two individuals that work together but only contact each other via online meetings would not be considered connected to each other.
Thus, there may be a small proportion of isolated vertices in the networks we consider, especially if we model situations when a community is under a disease control measure, like a stay-at-home order.
Since isolated vertices cannot infect others or be infected, they are irrelevant to the disease spreading dynamics in the network and we can exclude such vertices when constructing the graphs.
With proper book-keeping, isolated vertices have no impact on the epidemic results of the model, and one can easily put them back to the final result by adding them to the uninfected individuals.

For the vertices that are connected, we have the following definitions from graph theory:
\begin{definition}[Adapted from \cite{FriezeKaronski:2016}]
\label{def: Connected Graph}
Consider a simple, undirected graph $G=(V_G,E_G)$,
\begin{itemize}
\item A \textbf{path} within $G$ is a subgraph $P=(V_P,E_P)$ of $G$, where all vertices in $V_P$ can form a sequence such that every consecutive pair of vertices in the sequence is connected by an edge in $E_P \subseteq E_G$.
Moreover, $E_P$ only contains edges in $E_G$ that connect such consecutive pairs in the sequences. 

\item A \textbf{connected graph} $F=(V_F, E_F)$ is a graph wherein every pair of vertices $(a, b)$ where $a, b \in V_F$, there exist a path within $F$ such that its sequence starts with $a$ and ends with $b$.

\item A \textbf{component} $C$ of graph $G$ is a connected subgraph of $G$ that is \textbf{maximal}, which means no connected subgraph of $G$ will contains $C$ except $C$ itself.

\item A \textbf{cycle} is a path that contains no less than 3 vertices, whose vertices can form a sequence satisfying the requirement of path, while only the first and last vertices are the same.

\item A \textbf{tree} is a connected graph that contains no cycles.

\item A \textbf{tree component} is a component that contains no cycle, thus also a tree.
\end{itemize}
\end{definition}

\subsection{Random graphs}
\label{sec: randomgraphs}

In applications, graphs are often created directly as an accurate but abstract replica of the real-world network or data. 
In these cases the graph is mostly specified with clear topological structure.
This is possible only because the vertices and edges represented by these models are either physically fixed, like power cables, storage warehouses and highways, or recorded automatically as data by some system, like social media accounts, internet interactions and transportation records. 
As a result, these data are easily accessible and easy to track, and the graph can be practically created and modified accordingly. 

However, as described in section \ref{sec: basic}, epidemic network models have individuals in the community as vertices and disease-causing connections as edges.
Thus, even if for a small community the individuals are somehow tractable, but the disease-causing connections (even just typical social connections) can be difficult to ascertain.
Instead, it is common to approach modeling epidemic network models using random graphs.
Here we only discuss simple, undirected graphs in this manuscript so we now omit the adjectives ``simple'' and ``undirected'' and refer only to graphs for brevity.

Random graph theory is a vast subject. 
At the core are graphs that are generated randomly according to some certain graph properties. 
The goals are to find and discuss properties that would be shared among such graphs. 

\begin{definition}[\cite{Lovasz:2012}]
\label{def: graph property}
    A \textbf{graph property} is a property of graphs that depends only on the abstract structure, not on graph representations such as particular labelling or drawings of the graph.
\end{definition}

\begin{definition}[\cite{FriezeKaronski:2016}]
\label{def: property subset}
For a given graph property and given vertex set $[N]=\{1,2,\cdots ,N\}$, $N \in \mathbb{Z}_{>0}$, the \textbf{graph property subset} $\mathcal{P}_N$ is a subset of all labeled graphs satisfy the property,  i.e.  $\mathcal{P}_N=\{G=\{[N],E\}| G \text{ satisfies the given graph property}\}$.
\end{definition}

\begin{definition}
\label{def: Random Graph}
Let $N \in \mathbb{Z}_{>0}$.
Consider all graphs with size $N$ and a given graph property characterized by $\sigma$, which could be a single variable or a vector variable.
In this case we denote the graph property subset by $\mathcal{G}_{N,\sigma}$.
A \textbf{random graph} with size $N$ and the given property, denoted by $\mathbb{G}_{N,\sigma}$, is a graph chosen uniformly from $\mathcal{G}_{N,\sigma}$ at random.   
\end{definition}

The most well-known example of random graphs is the Erd\"os-R\'enyi graphs which were first studied by \cite{ErdosRenyi:1960} and \cite{Gilbert:1959}.
\begin{example}[Erd\"os-R\'enyi Graphs]
\label{ex: erdosrenyi}
The Erd\"os-R\'enyi graphs is parameterized by parameters $N \in \mathbb{Z}_{>0}$ and $p\in(0,1)$, and usually is labeled  $\mathcal{G}_{N,p}$. Any random graph $\mathbb{G}_{N,p} \in \mathcal{G}_{N,p}$ is constructed in the following manner: $N$ vertices are generated by assuming each of the $\binom{N}{2}$ possible pairs of vertices has an edge with probability $p$ independent of all other edges.
\end{example}

From Example~\ref{ex: erdosrenyi}, we saw that for any simple graphs $G$ with size $N \in \mathbb{Z}_{>0}$, there is at most $\binom{N}{2}$ edges in its edge set.
Moreover, if the size $N$ is fixed, there exist exactly $2^{\binom{N}{2}}$ distinguishable simple graphs.
For convenience, for given fixed size $N$, we use $[\binom{N}{2}]$ to denote the set of all possible edges for simple graphs, and use $[2^{\binom{N}{2}}]$ to denote the set of all possible simple graphs.
So $\forall G=\{[N],E\} \in [2^{\binom{N}{2}}]$, while $\forall E \subseteq [\binom{N}{2}]$ and for any graph property, its graph subset $\mathcal{P}_N \subseteq [2^{\binom{N}{2}}]$.

We are often interested in threshold phenomena, which is an ``abrupt'' appearance or disappearance of certain properties, as the size and number of edges of the graph increase monotonically.

\begin{definition} [Adapted from \cite{FriezeKaronski:2016}]
    A graph property with corresponding subset $\mathcal{P}_N$ is \textbf{monotone increasing property} if $G=\{V_G,E_G\} \in \mathcal{P}_N$ implies $G_{+e}=\{V_G, E_G \cup \{e\}\} \in \mathcal{P}_N$, i.e., adding an edge $e$ to a graph $G$ does not destroy the graph property.
\end{definition}

\begin{definition} [Adapted from \cite{FriezeKaronski:2016}]
    A monotone increasing property with corresponding subset $\mathcal{P}_N$ is \textbf{non-trivial} if for large enough size $N \in \mathbb{Z}_{>0}$, the empty graph (with no edges) $\bar{K}_N=\{[N], \emptyset \} \notin \mathcal{P}_N$ and the complete graph $K_N=\{[N],[\binom{N}{2}]\} \in \mathcal{P}_N$.
\end{definition}

It is equivalent to say that a monotone increasing property is non-trivial, if for large enough $N \in \mathbb{Z}_{>0}$, there exist at least one graph of size $N$ belonging to $\mathcal{P}_N$ and at least one graph of size $N$ that does not belong to $\mathcal{P}_N$.
A simple example of such a property is the existence of a loop. This is a monotone increasing property since adding edges to a graph with a loop will not destroy any existing loops.
In addition, for any $N \geq 3$ the empty graph $\bar{K}_N$ does not have any loops and the complete graph $K_N$ has loops.

\begin{definition} (Adapted from \cite{FriezeKaronski:2016})
\label{def: Threshold Function}
Consider random graph $\mathbb{G}_{N,\sigma}$.
A function ${\sigma}^*={\sigma}^*(N)$ of $N$ for the variable $\sigma$ is a \textbf{threshold} for an monotone increasing graph property with graph property subset $\mathcal{P}_N$ of the random graph $\mathbb{G}_{N,\sigma}$ if,
    \begin{equation}
    \lim_{N \rightarrow \infty} \mathbb{P} (\mathbb{G}_{N,{\sigma}} \in \mathcal{P}_N) =
        \begin{cases}
            0, & \text{if } {\sigma} < {\sigma}^*\\
            1, & \text{if } {\sigma} > {\sigma}^*\\
        \end{cases}
    \end{equation}
\label{defn:threshold}
\end{definition}

\begin{example}[\cite{FriezeKaronski:2016}]
\label{exp: threshold}
Consider the Erd\"os-R\'enyi random graph from Example~\ref{ex: erdosrenyi}.
Let the characterizing parameter $\sigma=p$, the probability that an edge of the graph exists (independent of all other edges).
Now consider the property that the graph has at least one edge, i.e. the graph is non-empty graph.
Now $\mathcal{P}_N=\{G=\{[N],E\}|E \subseteq [\binom{N}{2}] \text{ and } E \neq \emptyset \}=[2^{\binom{N}{2}}] \setminus \{\bar{K}_N\}$.
This simple graph property is clearly monotone increasing. It was proven that $p^{*}= 1/N^2$ is the threshold for $\mathbb{G}_{N,p}$ to have at least one edge ($\mathbb{G}_{N,p} \neq \bar{K}_N $). 
\end{example}

In order to clarify the statements of many theorems related to random graphs in the limit as the size goes to infinity, we follow \cite{FriezeKaronski:2016} and use the notation that an event occurs ``with high probability'' (w.h.p.) if it happens almost surely in the limit as the size goes to infinity. This is given in the next definition.

\begin{definition}[\cite{FriezeKaronski:2016}]
\label{def: whp}
We say a sequence of events $\{E_N\}$ occurs \textbf{with high probability (w.h.p.)} if
$
\lim_{N \rightarrow \infty} \mathbb{P}(E_N)=1
$.
\end{definition}

Thus in Definition~\ref{defn:threshold} the value ${\sigma}^*$ is a threshold for a property $\mathcal{P}_N$ in $\mathbb{G}_{N,\sigma}$ is the same as saying that $\mathbb{G}_{N,\sigma}$ has property $\mathcal{P}_N$ w.h.p. if $\sigma<\sigma^*$, while it is false w.h.p. if $\sigma>\sigma^*$.
\cite{BollobasThomason:1987} found that thresholds are significant for random graphs. The following theorem makes this clear.
\begin{theorem}[\cite{BollobasThomason:1987}]
\label{thm: Threshold}
    Every non-trivial monotone graph property has a threshold.
\end{theorem}
\begin{proof}
See proof by \cite{FriezeKaronski:2016}.
\end{proof}

Since we will be working on epidemic network models of large communities, these graphs are usually assumed to have large enough size, such that the results at the $N \rightarrow \infty$ limit are reasonable to use.
We will show later that these thresholds in random graph theory is highly related to, and eventually equivalent to the epidemic threshold that determines the appearance of epidemic outbreak.

For large graphs, especially as we take the limit as $N\to\infty$, it can be convenient to model the graphs using degree distributions of vertices.
We define these distributions below for a fixed graph $G$, then explain later on how these degree distributions can be used to generate random graphs.

\begin{definition}
\label{def: Degree Distribution}
For a graph $G$ with no isolated vertices and for any positive integer $k \in \mathbb{Z}_{>0}$, let $p_k \in (0,1)$ be the probability that a randomly chosen vertex have degree $k$. 
Consider the degree $K$ of a randomly chosen vertex from $G$ as a random variable, then the \textbf{degree distribution} of $G$ is the probability distribution of $K$ with the probability mass function given by,
    \begin{equation}
        p_k=\mathbb{P}(K=k).
    \end{equation}
\end{definition}

Several commonly used discrete probability distributions are (a) mean field $k=const.$, (b) discrete Poisson $p_k=e^{-\lambda} \lambda^k/k!$, (c) discrete exponential $p_k=(1-e^{-\alpha}) e^{-\alpha (k-1)}$ and (d) scale-free (also known as truncated power-law) $p_k=k^{-\gamma}/\zeta(\gamma)$ \citep{BansalMeyers:2012}.

For a graph with finite size $N$, it is common  to specify the degree distribution using a degree sequence:
\begin{definition}
\label{def: Degree Seq}
For a graph $G$ with size $N \in \mathbb{Z}_{>0}$ and no isolated vertices, consider an ordered sequence of all its vertices labeled by $1, 2,\cdots, N$.
The \textbf{degree sequence} of the graph is a sequence $\textbf{d}=(d_1, d_2, \cdots, d_N)$, such that $d_i \in \mathbb{Z}_{>0}$ and degree of vertex $v_i$ is given by $D_G({v_i})=d_i$ for any $i=1,2,\cdots,N$.
The set of all such simple graphs is denoted by $\mathcal{G}_{N,\textbf{d}}$ while corresponding random graph is denoted by $\mathbb{G}_{N,\textbf{d}}$.
\end{definition}

It is easy to see that for a graph with finite size $N$, specifying the degree sequence is equivalent to specifying the degree distribution, by setting $p_k=|\{i|d_i=k, i=1,2,\cdots,N \}|/N$ for all $k \in \mathbb{Z}_{>0}$.
Defining the probabilities $p_k$ in this way guarantees that it is properly normalized and $\sum_{k=1} p_k=1$. For a finite graph $p_k$ is also the portion of vertices with degree $k$ among all vertices.

Even though a lot of the results we will discuss are based on the asymptotic property of graphs derived in the large $N$ limit ($N \rightarrow \infty$), all practical realizations of graphs still have finite sizes and represent finite populations.
Degree sequences are especially useful for both the applications, when vertices and their degrees are directly generated from data, and simulations, which use degree sequences to generate random networks.
Even in cases when we have a prescribed distribution, these are often derived from a sample result with a different sample size or come from an assumption that follow some classical probability distribution.
These prescribed distributions often cannot be directly used as the proportion of vertices for each positive integer value $k$ since they might not lead to integer number of vertices.
For these cases, the prescribed distribution can instead be used to  generate a degree sequence by taking each $d_i$ as i.i.d. variables from the prescribed distribution, so the final degree distribution is approximated by the described distribution for large enough $N$. 
Therefore, it is more convenient to use the degree sequence to specify random graphs, then generate the degree distribution based on such sequences for later analysis and calculation steps of the model.

\begin{definition}
\label{def: Degree Seq Addition}
Consider a degree sequence $\textbf{d}=(d_1, d_2, \cdots, d_N)$ with size $N$.
We define the following:
\begin{itemize}
    
\item The \textbf{maximum degree} $\Delta = \max\{d_i|i=1,2,..,N\}$.

\item The \textbf{total degree} $M = \sum_{i=1}^{N} d_i$.

\item The \textbf{mean degree} $\Bar{d}=\langle K \rangle=M/N$.

\item The degree sequence $\textbf{d}$ is said to be \textbf{realizable} if the set of simple graphs $\mathcal{G}_{N,\textbf{d}} \neq \emptyset$. 
    
\item For a degree sequence $\textbf{d}$, any random graph $\mathbb{G}_{N,\textbf{d}} \in \mathcal{G}_{N,\textbf{d}}$ is a \textbf{realization} of $\textbf{d}$.
\end{itemize}
\end{definition}

We note that there can be differences between the prescribed distribution of the degree sequence and the exact degree distribution of the graph.
For finite size $N$ graphs, there will be increasing difference between them as $N$ gets smaller.
For large $N$, the generating distribution is a close approximation of degree distribution.

\begin{example}
The Erdos-Renyi graphs generated using the algorithm described in Example~\ref{ex: erdosrenyi} results in a Binomial($N-1,p$) degree distribution with $p_k=\binom{N-1}{k}p^k(1-p)^{N-1-k}$ for $k=0,\dotsc,N-1$.
In the large $N$ limit, this degree distribution can be approximated by a Poisson distribution.
We note that the construction of the Erd\"os-R\'enyi graph described in Example~\ref{ex: erdosrenyi} is straightforward and does not include any random draws from the binomial degree distribution. 
However, it turns out that many real-world graphs, such as those from social networks, have degree distributions that are quite different from Poisson~\citep{NewmanStrogatzWatts:2001}.\end{example}

Given a degree distribution, the next step would be generating the random graph uniformly.
There are different ways to describe and generate random graphs $\mathbb{G}_{N,\textbf{d}}$ with size $N$ a prescribed degree sequence $\textbf{d}=(d_1,d_2,...,d_N)$.

A common confusion about some of the early work on network epidemiology \citep{NewmanStrogatzWatts:2001,Newman:2002,Newman:2003} is that these articles provide little detail about the generation algorithm they used to generate random networks and do not provide explicit statistical interpretation on how these algorithms guarantee uniform graph outputs.
In particular, there is no discussion about how to guarantee uniformity while rejecting loops and multi edges.
The intuitive description from those papers suggest that the algorithm is similar to that presented by \cite{BlitzsteinDiaconis:2011}. 
This algorithm does indeed have high efficiency and is guaranteed to generate all possible graphs, but it remains open to prove that the output is uniformly distributed.

It is still an active research topic in random graph theory \citep{Greenhill:2021} to develop algorithms that efficiently generates all random graphs with the given degree sequence under a large size $N$. It is also difficult to guarantee uniformity.
All existing methods have their limitation, like calculation speed, extra requirement of degree sequence, case by case code interpretation and cannot analytically guarantee uniform output.

To systematically illustrate the generation method, we would like to use the famous configuration model developed by \cite{Bollobas:1980}.
In Algorithm~\ref{alg: Config} we provide a short description of an algorithm based on configuration model, with some slight modification to make it more consistent with both later theorems and Newman's terminology.

\begin{algorithm}
\caption{Construction of a random network from a degree sequence with configuration model}
\label{alg: Config}
Suppose we want to generate a realization of a finite graph with $N$ vertices and a given degree sequence $\textbf{d}=(d_1,d_2,...,d_N)$. 

\begin{enumerate}
\item Following the algorithm described by \cite{Newman:2010}, we first create $N$ vertices labeled from $i=1$ to $N$, and attach $d_i$ \textbf{stubs} (half edges) to the $i$-th vertex.

\item Then we randomly pair any of the two distinct stubs to form edges by the configuration model: 
\begin{enumerate}
\item Vertices are treated as a partition of all stubs, which each vertex labeled by $i$ is isomorphic to a set $w_i$ of the partition contains all the stubs attach to it, so the set $w_i$ has $d_i$ element.

\item Randomly picking any two of the distinct stubs among all stubs and form a pair, even they come from the same partition set or there are already other pairs between the same two sets.
Once two stubs are paired, they will be recorded and removed from later paring candidates.

\item Repeat Step (b) until no unpaired stubs are left.

\item Map the paired stubs and partition sets back to graph, as edges and vertices, accordingly.

\item If there are multi-edges and loops in the result graph reject the result and restart from (a), until a simple graph is attained.
\end{enumerate}

\item Output the simple graph as final result.
\end{enumerate}
\end{algorithm}

A more systematic and detailed algorithm to randomly pairing the stubs by configuration model is given by \cite{FriezeKaronski:2016}, with analytical proof that the resulting graph is guaranteed to be uniformly selected at random from among all possible graphs.
The configuration model is the earliest algorithm that guarantee uniformity and is easy to interpret, but it requires rejection of loops and multi-edges after each whole generation procedure and restart until success, thus is slow for large $N$.
We will instead introduce another algorithm for simulation and validation of the project.

With the pairing procedure described in Algorithm~\ref{alg: Config}, it is obvious that not every positive integer sequence is a realizable degree sequence.
Moreover, the following must be true:
\begin{corollary} 
\label{cor: realizable seq}
For any realizable degree sequence $\textbf{d}=(d_1, d_2, \cdots, d_N)$ with size $N$
    \begin{itemize}
        \item The total degree $M= \sum_{i=1}^{N} d_i$ is even.

        \item For a realization $\mathbb{G}_{N,\textbf{d}}$ of $\textbf{d}$, the total number of edges $m$ must satisfy $ m = \frac{1}{2} M$
    \end{itemize}
\end{corollary}

This is the only condition that \cite{NewmanStrogatzWatts:2001} and \cite{MillerSlimVolz:2012} mention for a degree sequence. However, there are more conditions needed for a degree sequence to be realizable.

\begin{theorem}[\cite{ErdosGallai:1960}]
\label{thm: Realizable}
    Let $d_1 \geq d_2 \geq \cdots \geq d_N$ be non-negative integers with $M=\sum_{i=1}^{N} d_i$ even. Then the sequence $\textbf{d}=\{d_1,d_2,\cdots,d_N\}$ is realizable if and only if
    \begin{equation}
        \sum_{i=1}^{k} d_i \leq k (k-1) + \sum_{i=k+1}^{n} \min (k,d_i) \text{ for each } k \in \{1,2,\cdots,N\}
    \end{equation}
\end{theorem}
For any subset of $k$ vertices, the left side of the inequality is the maximum number of edges that these vertices could have, since it is the sum of the first $k$ largest degrees in $\mathbb{d}$.
The right side of the inequality comes from another counting estimation.
For such $k$ vertices in the subset, there are at most $\binom{k}{2}=k(k-1)$ edges that connected to each other within the subset.
For each vertex $a$ outside the subset, there are at most $\min (k,D_G(a))$ edges connected to vertex in the subset, where such number is bounded by $\sum_{i=k+1}^{n} \min (k,d_i)$.
This clearly shows the necessity of the condition, but sufficiency is less  obvious and we refer to \citep{ErdosGallai:1960} for further discussion.

Given a degree sequence generated from a given degree distribution, to guarantee the sequence is realizable we have the following additional algorithm that we need to apply before generating a random graph:

\begin{algorithm} 
\caption{Generation of realizable degree sequence from a prescribed distribution}
\label{alg: Seq}
Consider size $N$ and the prescribed distribution given by $\mathbb{P}(D=d), d\in \mathbb{Z}_{>0}$

\begin{enumerate}
\item Following the algorithm described by \cite{Newman:2010}, we first draw a set of $N$ positive integer from the prescribed distribution.

\item If the sum of all $N$ integers in Step 1 is even, proceed the sequence to Step 3. 
If there is an odd total Degree, randomly chose one of the vertices and changed its degree by drawing again from the prescribed distribution until the total number of stubs is even, then proceed to the next step.

\item Since both the order of vertices and the degree sequence have no impact when generating random graphs, we re-order and sort the sequence in descending order, as degree sequence $\textbf{d}$.

\item Check if the sequence $\textbf{d}$ satisfies Theorem~\ref{thm: Realizable}.
If satisfied, take $\textbf{d}$ as degree sequence.
Otherwise, restart from Step 1.
\end{enumerate}
\end{algorithm}

\subsection{Components}
\label{sec: components}
Taking the limit when the size of the graph $N\to\infty$ allows us to model large networks with a lot of the analysis simplified.
It is important to recognize entities that are called extensive or giant components. 
These are components that contain a significant proportion of the overall number of vertices of the graph and has a much larger size compared to other components in the same graph.
It is not hard to see that adding edges to a graph with an existing giant component does not break the existence of a giant component with a dominant size compared to other components.
Also, an empty graph has all its components with the same size of one, while complete graph contains only one single giant component with size $N$.
Therefore, the appearance and existence of giant appearance is a nontrivial monotone increasing graph property of random graphs. Thus, by Theorem~\ref{thm: Threshold}, it must have a threshold.

\begin{example}
The existence and uniqueness of such giant component was originally observed in infinite Erd\"os-R\'enyi random graphs $\mathbb{G}_{N,p}$ by \cite{ErdosRenyi:1960}, where the threshold is given by $p^{*}=1/N$.
More precisely, for large enough network size $N$,
    \begin{enumerate}
        \item If $p<p^{*}=1/N$, w.h.p. $\mathbb{G}_{N,p}$ has no component of size larger than $O(\log(N))$.
        \item If $p=p^{*}=1/N$, w.h.p. $\mathbb{G}_{N,p}$  has a unique component of size proportional to $N^{\frac{2}{3}}$.
        \item If $p>p^{*}=1/N$, w.h.p. $\mathbb{G}_{N,p}$ has a unique giant component of size $\Theta N$, which contains a positive fraction $\Theta \in (0,1)$ of all $N$ vertices. 
        In addition, w.h.p. no other component other than the giant one will contain more than $O(\log(N))$ vertices.
    \end{enumerate}
\end{example}

It can be difficult to define a giant component in random graphs outside of the Erd\"os-R\'enyi system.
However, such an extensive component remains particularly important in many network applications.
For example, in the graph corresponding to the Internet, if there was no giant component the network would not be able to perform its intended role.
To discuss giant components for arbitrary degree sequences we first consider the following result from \cite{MolloyReed:1998}.

\begin{theorem}(Adapted from \cite{MolloyReed:1995})
\label{thm: Giant Component Threshold}
    Consider random graphs $\mathbb{G}_{N,\textbf{d}}$ with size $N$ and realizable degree sequence $\textbf{d}$.
    For $\textbf{d}$, let $\Delta$ denotes its maximum degree, $\langle K \rangle$ denotes its mean degree and the corresponding degree distribution be given by a properly normalized probability mass function $\mathbb{P}(K=k)=p_k$ for all $k \in \mathbb{Z}_{>0}$.
    Suppose $\Delta= O(N^{\frac{1}{4}-\epsilon})$ for some $\epsilon > 0$ and let $\Lambda=\Lambda(d)=\sum_{k=1} k (k-2) p_k$.
    Then:
    \begin{enumerate} 
        \item[(a)] If $\Lambda > 0$, then w.h.p. there exist constants $\Theta, \zeta$ that depends on $\textbf{d}$ such that $\mathbb{G}_{N,\textbf{d}}$ has a giant component with size at least $\Theta N$ and with at least $\zeta N$ cycles.
        Furthermore, w.h.p. this component is the only component of \textbf{$\mathbb{G}_{N,\textbf{d}}$} with size greater than $\gamma \log (N)$ for some constant $\gamma$ depend on $\textbf{d}$.

        \item[(b)]  If $\Lambda < 0$, and for some function $\omega(N) = O(N^{\frac{1}{8}-\epsilon})$, $\Delta \leq \omega(N)$, then w.h.p. for some constant $R$ depending on $\Lambda$ (and thus also on $\textbf{d}$), $\mathbb{G}_{N,\textbf{d}}$ has no component with more than $R \omega(N)^2 \log(N)$ vertices.
        Also, w.h.p. $\mathbb{G}_{N,\textbf{d}}$  a.s. has fewer than $2 R \omega(N)^2 \log(N)$ cycles and has no component containing more than one cycle.   
    \end{enumerate}
\end{theorem}

Since the idea of giant component is originally defined for infinite Erd\"os-R\'enyi graph, it is difficult to give a precise definition of a giant component for a fixed finite network \citep{KissMillerSimon:2017}.
But to better discuss such as key property of the network model, we provide a definition for network with large enough size $N$ following the idea of Theorem~\ref{thm: Giant Component Threshold} with some level of ambiguity.

\begin{definition}
\label{def: giant component}
For graph with a large enough size $N$, if the component with largest size has size greater than $\gamma \log(N)$ with some constant $\gamma$ bounded away from zero, such component is said to be a \textbf{giant component.}
Any component that is not a giant component is called a \textbf{small component}.
\end{definition}

In the network models we are interested in, the degree represents number of disease-causing connections.
The maximum number of such connections is likely to be limited and bounded, even for large population sizes.
This is also applicable for the distributions previously mentioned in Section~\ref{sec: randomgraphs}, which are commonly used to generate degree sequence.
With proper parameters, the maximum degree of the sequence generated by these distributions are either bounded (as in the Mean Field case) or can be truncated.
With that in mind, we assume that the maximum degree $\Delta$ always satisfies $\Delta=O(1)$, so the condition in Theorem~\ref{thm: Giant Component Threshold} is satisfied for large enough $N$.

\cite{MolloyReed:1998} provide further result based on Theorem~\ref{thm: Giant Component Threshold} to give a size estimate of the giant component. 
Based on that, \cite{FriezeKaronski:2016} proved special case under the condition $\Delta=O(1)$ for the same threshold $\Lambda=\Lambda(d)=\sum_{k=1} k (k-2) p_k$ with more details.
\begin{theorem}[\cite{FriezeKaronski:2016}]
\label{thm: Giant Component Size}
    For $\textbf{d}$, with maximum degree $\Delta = O(1)$ and the its corresponding degree distribution is given by a properly normalized probability mass function $\mathbb{P}(K=k)=p_k$ for all $k \in \mathbb{Z}_{>0}$ such that $\sum_{k=1} p_k=1$.
    Let $\Lambda=\Lambda(d)=\sum_{k=1} k (k-2) p_k$ and $N \rightarrow \infty$ large enough.
    \begin{enumerate} 
        \item[(a)] If $\Lambda < 0$, then w.h.p. the size of the largest component in $\mathbb{G}_{N,\textbf{d}}$ is $O(\log N)$, thus all component is small component.
        
        \item[(b)] If $\Lambda > 0$, then w.h.p. there is a unique giant component of linear size $\Theta N$ where $\Theta \in (0,1]$ is defined as follows:
        Let
        \begin{equation}
        \label{eqn: alpha eqn}
            f(\alpha)=\langle K \rangle - 2 \alpha - \sum_{k=1}^{\Delta} k p_k \Big(1-\frac{2 \alpha}{\langle K \rangle}\Big)^{k/2}
        \end{equation}
        and $\psi$ be the smallest positive solution to $f(\alpha)=0$.
        Then
        \begin{equation}
        \label{eqn: size Theta eqn}
            \Theta=1-\sum_{k=1}^{\Delta} p_k \Big(1-\frac{2 \psi}{\langle K \rangle}\Big)^{k/2}
        \end{equation}
        If $p_1=0$ then $\Theta=1$ otherwise $\Theta \in (0,1)$.

        \item[(c)] In Case(b), the degree sequence of the graph obtained by deleting the giant component satisfies the conditions of Case (a).
    \end{enumerate}
\end{theorem}
Clearly, $\Lambda$ is the threshold parameter for the existence of giant component. 
There is no conclusion at the transition point $\Lambda=0$.
If we treat $\Theta$ as the proportion of the vertices that belongs to a giant component and let Case (a) in Theorem~\ref{thm: Giant Component Size}, where no giant component exists as $\Theta=0$, a natural corollary would be:
\begin{corollary}
\label{cor: component structure}
    For random graph $\mathbb{G}_{N,\textbf{d}}$ satisfied condition of Theorem~\ref{thm: Giant Component Size},
    If $\Lambda \neq 0$, one and only one of the following three cases must be true w.h.p. 
    \begin{enumerate} 
        \item[(a)] The random graph $\mathbb{G}_{N,\textbf{d}}$ contains only one component, which is the giant component with $\Theta=1$.

        \item[(b)] The random graph $\mathbb{G}_{N,\textbf{d}}$ contains only one giant component with $\Theta \in (0,1)$ and all of the rest components are small components w.h.p..

        \item[(c)] The random graph $\mathbb{G}_{N,\textbf{d}}$ contains no giant component as $\Theta=0$, and all components will be small components.
    \end{enumerate}
\end{corollary}

Moreover, we have the following result:
\begin{theorem}
\label{thm: component}
 Under the conditions of Theorem~\ref{thm: Giant Component Size} and $\Lambda \neq 0$,
    \begin{enumerate}
        \item[(a)] If it exists, the giant component in $\mathbb{G}_{N,\textbf{d}}$ is unique and contains cycles w.h.p..
        
        \item[(b)] Any small components in $\mathbb{G}_{N,\textbf{d}}$ must be a tree w.h.p..
    \end{enumerate}
\end{theorem}

\begin{proof}
    (a) is a clear inference from Corollary~\ref{cor: component structure} and Theorem~\ref{thm: Giant Component Threshold} (a)

    (b) For the case in Corollary~\ref{cor: component structure} (c) that $\Lambda<0$, $\mathbb{G}_{N,\textbf{d}}$ only contains small component, consider Theorem~\ref{thm: Giant Component Threshold} (b), the maximum size of all small component is $R \omega(N)^2 \log(N)$.
    Thus the number of components is at least:
    \begin{equation}
        \frac{N}{R \omega(N)^2 \log(N)}
    \end{equation}
    Similarly, since $\mathbb{G}_{N,\textbf{d}}$  has fewer than $2 R \omega(N)^2 \log(N)$ cycles w.h.p. and has no component that contains more than one cycle w.h.p., the number of non-tree components (which are also the components with one cycle) is then less than $2 R \omega(N)^2 \log(N)$.
    With $\omega(N)=O(N^{\frac{1}{8}-\epsilon})$, we have
    \begin{align}
        \mathbb{P}(\text{small component of } \mathbb{G}_{N,\textbf{d}} \text{ has cycle}) 
        & =\frac{\text{number of non-tree small component}}{\text{number of small component}}
        \nonumber
    \\ 
        & \leq \frac{2 R \omega(N)^2 \log(N)}{\frac{N}{R \omega(N)^2 \log(N)}} = 2R^2 \frac{ \omega^4 \log(N)^2}{N} 
        \nonumber
    \\
        & \leq 2 R^2 \frac{(N^{1/8})^4 \log(N)^2}{N} = 2 R^2 \log(N)^2 N^{-\frac{1}{2}} 
        \nonumber
    \\
        & \to 0 \text{ as } N\rightarrow \infty
    \end{align}
    Therefore,
    \begin{align}
        &\lim_{N \rightarrow \infty} \mathbb{P}(\text{small component in }\mathbb{G}_{N,\textbf{d}} \text{ is tree})
        \nonumber
        \\
        =& 1-\lim_{N \rightarrow \infty} \mathbb{P}(\text{small component in } \mathbb{G}_{N,\textbf{d}} \text{ has cycle})=1
    \end{align}
    For the case in Corollary~\ref{cor: component structure} (b) with $\Lambda>0$ and $\Theta \neq 1$, the result follows from applying the previous case result with Theorem~\ref{thm: Giant Component Size} (c).
\end{proof}

Theorem~\ref{thm: Giant Component Size} also indicates that even with the randomness in generating process, the component structure of the random graph $\mathbb{G}_{N,\textbf{d}}$ is determined w.h.p. by its degree sequence $\textbf{d}$.
From the probability point of view, the expectation of the size of the giant component among all random graphs $\mathbb{G}_{N,\textbf{d}}$ with given $\textbf{d}$ is a good approximation of the size of the giant component for random graphs.
This theorem provides a solid justification of the methods discussed later for network epidemiology based on expectations of random graphs.


\subsection{Important generating functions and basic results} 
\label{sec: pgfs}



To study the spread of disease in a network, we would like to use similar approaches as in Section~\ref{sec: components}, but we need to allow more ``degrees of freedom'' to accommodate randomness in disease spreading.
Recall that Theorem~\ref{thm: Giant Component Threshold} and Theorem~\ref{thm: Giant Component Size} proved that the component structure of random graph $\mathbb{G}_{N,\textbf{d}}$, especially the giant component size, are determined w.h.p. by the network size $N$ and the threshold $\Lambda$.
Thus the expectation of the component size is a good estimate of the component size for any realization.
In this section, we use probability generating functions (PGFs) to find the expectation of these component sizes for network with large size $N$ and prove that such expectation provides an equivalent estimate as in Theorem~\ref{thm: Giant Component Size}. This PGF approach will also give a threshold parameter for the existence of giant component equivalent to the threshold using $\Lambda$.
This will allow us to extend the ideas later on  in Section~\ref{sec: perco} to accommodate transmissibility of diseases to the PGF method.

We always assume in this section that we have a random graph $\mathbb{G}_{N, \textbf{d}}$ with large enough size $N \rightarrow \infty$, generated from a degree sequence $\textbf{d}$ and corresponding degree distribution $p_k$.
Recall that if $K$ is the degree of a randomly chosen vertex in the network we use the notation $p_k=P(K=k)$ for its probability mass function. The corresponding probability generating function is given by,
\begin{equation}
\label{eqn: pgfgp}
    G_p(x)=\sum_{k=1}^\infty p_k x^k.
\end{equation}
Properties of PGFs are reviewed in Appendix~\ref{sec: pgf}.

Recall the random graph generation algorithm in Algorithm \ref{alg: Config}. We introduce the following definitions about excess degree and sub-component of randomly chosen stub.
\begin{definition}
\label{def: Excess Degree}
    Suppose from any graph $G=(V_G,E_G)$, randomly choose an edge $e \in E_G$ and then randomly choose one of the two stubs that forms the edge. Let the vertex $a $ be the vertex connected to that stub.
    Let $G'=(V_{G'},E_{G'})$ be a subgraph of $G$ such that $V_{G'}=V_G$ and $E_{G'}=E_G\setminus\{e\}$.
    \begin{itemize}
        \item \textbf{Excess degree} of the randomly chosen stub is the number of stubs (edges) attached to $a$, excluding the chosen stub (edge) and is thus equal to $D_{G'}(a)=D_G(a)-1$.

        \item \textbf{Sub-component} of the randomly chosen stub is a subgraph of $G$, such that it is also the component of $G'$ that contains $a$.
    \end{itemize}
\end{definition}

\begin{theorem} [Adapted from \cite{Newman:2002}]
\label{thm: pgfGq} 
    The PGF corresponding to excess degree of a randomly chosen stub in random graph $\mathbb{G}_{N, \textbf{d}}$ is given by,
    \begin{equation}
        G_q(x) =\frac{G_p'(x)}{G_p'(1)}.
    \label{eqn: pgfGq}
    \end{equation}
\end{theorem}

\begin{proof}
    Suppose that we randomly choose an edge and then randomly choose one of the two vertices connected to that edge. 
    We denote this vertex by $a$ and note that a vertex with degree $k$ is $k$ times as likely to be chosen in this manner than a vertex with degree $1$.
    Thus the probability mass function of the excess degree of $a$ is given by,
    \begin{equation}
    \label{eqn: qk}
        q_{k-1}=P(\text{excess degree}=k-1) = \frac{k p_k}{\sum_{k=1}^{\infty}k p_k} =\frac{k p_k}{\langle K\rangle}. 
    \end{equation}
From this and Theorem~\ref{thm: pgf}, the PGF corresponding to the ``excess degree distribution'' is given by,
\begin{equation}
G_q(x)=\sum_{k=0}^\infty q_k x^k=\sum_{k=0}^\infty \frac{(k+1)p_{k+1}x^k}{{\langle K\rangle}}=\sum_{k=1}^\infty \frac{k p_{k}x^{k-1}}{{\langle K\rangle}}=\frac{G_p'(x)}{G_p'(1)}.
\end{equation}
\end{proof}

Now recall from Corollary~\ref{cor: component structure} that for a random graph $\mathbb{G}_{N,\textbf{d}}$ with large enough size $N \rightarrow \infty$ and $\Lambda \neq 0$, the component-level structure of the graphs only has three possible cases.
Let $\mathcal{S} \in [0,1]$ denote the proportion of vertices that belongs to the unique giant component. Then the possibiliies are:

\begin{enumerate}
    \item The graph contains only one giant component with $\mathcal{S}=1$. 
    No small component exists.

    \item The graph contains one unique giant component with $\mathcal{S} \in (0,1)$ and many small components.

    \item The graph contains no giant component such that $\mathcal{S}=0$ and thus contains only small component.
\end{enumerate}

Component structure is clear for case (1), thus we focus more about the second and third case where $\mathcal{S} \in [0,1)$.
More specifically, we care more about the transition behavior around the threshold of giant component existence, where the small components start to merge into a giant component.
We note that
$\mathcal{S}$ can also been seen as the probability that a randomly chosen vertex of $\mathbb{G}_{N,\textbf{d}}$ belongs to giant component.
Recall from Theorem~\ref{thm: component} that the probability that a randomly chosen vertex $a$ of $\mathbb{G}_{N,\textbf{d}}$ belongs to small component is the same as the probability it belongs to a tree component.
\begin{align}
    &\mathbb{P}(a \in \text{small component})  =\mathbb{P} (a \in \text{tree component})
    \nonumber
    \\
    =&1-\mathbb{P}(a \in \text{giant component})=1-\mathcal{S}
\end{align}

Now consider $s \in \mathbb{Z}_{>0}$ and assume a randomly chosen vertex $v$ of $\mathbb{G}_{N,\textbf{d}}$ belongs to small component with size $s$ has probability mass function $P(s)$, such that:
\begin{align}
    P(s) & =\mathbb{P}(a \in \text{small component } \wedge \text{ the small component has size }s)
    \\
    & =\mathbb{P}(a \in \text{tree component } \wedge \text{ the tree component has size }s)
\end{align}

\begin{definition}
\label{def: Hp}
Assume that $\mathcal{S} \in [0,1)$.
The conditional probability that a randomly chosen vertex $v$ belongs to a component with size $s$, given that the vertex $v$ belongs to a small (tree) component will be given by $\hat{P}(s)$. Clearly,
\begin{equation}
\label{eqn: hatPs}
   \hat{P}(s) =\mathbb{P}(a \in \text{component with size } s | a \in \text{ a tree component }) = \frac{P(s)}{1-\mathcal{S}}
\end{equation}
Using the same idea as in the PGF defined in \eqref{eqn: pgfgp}, we define the corresponding PGF by
\begin{equation}
\label{eqn: pgfhp}
    H_p(x)=\sum_{s=1}^{\infty} \hat{P}(s) x^s=\frac{1}{1-\mathcal{S}} \sum_{s=1}^{\infty} P(s) x^s. 
\end{equation}
In addition, recall that we defined the PGFs of excess degree $G_q(x)$ for a randomly chosen stubs in Theorem~\ref{thm: pgfGq} based on PGFs of degree for randomly chosen vertex.
Since we have $H_p(x)$ for randomly chosen vertex, we can also define a PGF $H_q(x)$ corresponding to the conditional probability that a randomly chosen stub has a sub-component with size $s$, given the stub is connected to a tree sub-component.
\end{definition}

To derive an expression for $H_q(x)$, we let $u$ be the probability that a randomly chosen stub has a tree sub-component.
Since $\mathcal{S} \in[0,1)$ by Theorem~\ref{thm: component} we must have $u \neq 0$.
The expression for $H_q(x)$ is given by the following theorem:

\begin{theorem}
\label{thm: PGFHq}
For a random graph $\mathbb{G}_{N, \textbf{d}}$ with large enough $N\rightarrow\infty$ and contains a small component, the PGF $H_q(x)$ must satisfy the following equation:
    \begin{equation}
    \label{eqn: HqEqn}
        H_q(x)=\frac{x}{u} G_q(u H_q(x))
    \end{equation}
where $u$ is the probability that a randomly chosen stub has a tree sub-component.
\end{theorem}

\begin{proof}
Let $e$ be a randomly chosen stub and $a$ be the vertex connected to it. 
We refer to all other stubs (other than $e$) that are connected to $a$ as the excess stubs of $a$.
By definition, $a$ belongs to the sub-component of $e$.
We also see that $e$ has a tree sub-component if and only if all excess stubs of $a$ are paired with stubs that have tree sub-components.
As we can see in Algorithm~\ref{alg: Config}, all stubs are assigned independently to vertices under the degree distribution. 
Edges are formed by randomly pairing stubs, for large enough size $N$, this pairing procedure can be seen as independent at vertex level \citep{MillerSlimVolz:2012}, so for any $k \in \mathbb{Z}_{>0}$ we have:
\begin{align}
&\mathbb{P}(\text{Excess degree of e is } k | \text{a has tree sub-component} )
    \nonumber
\\
    = &\frac{\mathbb{P}(\text{Excess degree of } e \text{ is } k \text{ and all excess stubs of } a \text{ are paired with stubs with tree sub-components} ) }{ \mathbb{P}(a \text{ has tree sub-component})}
    \nonumber
\\
    = &\frac{q_k u^k}{u} = q_k u^{k-1}
    \label{eqn: CondProb}
\end{align}
where $q_k$ is the probability mass function of excess degree, defined in \eqref{eqn: qk}.

In what follows we assume that the randomly chosen stub $e$ is indeed connected to a tree sub-component.
In this restricted probability space, $H_q(x)$ is the PGF of the size of the sub-component of $e$. 
All excess stubs will also be connected to a tree sub-component.
Due to the construction of the random graph, the sizes of those tree sub-components will be independent and have the same PGF $H_q(x)$ (see Figure~\ref{fig: H1} for an illustration).
From property 4 in Theorem~\ref{thm: pgf}, the sum of the sizes of all the tree sub-components attached to the excess stubs has PGF given by, 
\begin{equation}
\sum_{k=0}^\infty q_k u^{k-1} (H_q(x))^k.
\label{eq: decomp}
\end{equation}
Since we are dealing with trees, by Theorem~\ref{thm: component} the size of the sub-component attached to the original stub will be equal to this sum plus one (to include the vertex $a$). From property 2 in Theorem~\ref{thm: pgf}, we have the following equation for the PGF $H_q(x)$,
\begin{equation}
H_q(x) = x\sum_{k=0}^\infty q_k u^{k-1} (H_q(x))^k
= \frac{x}{u}\sum_{k=0}^\infty q_k  (uH_q(x))^k
= \frac{x}{u}G_q (uH_q(x)).
\label{eqn: Hq_PGF}    
\end{equation}
This completes the proof of this theorem.
\end{proof}


\begin{figure}[htbp]
\begin{center}
    \includegraphics[width=0.7\textwidth]{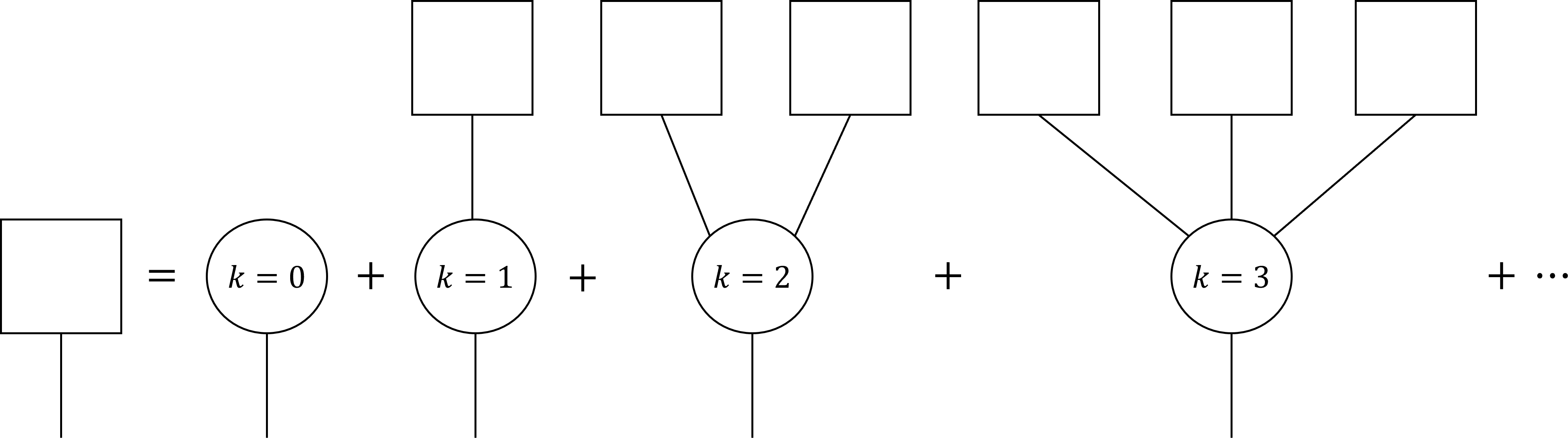}
\end{center}
\caption[Decomposition of $H_q(x)$]{This figure is based on the explanation and figure presented by ~\cite{NewmanStrogatzWatts:2001} to explain the derivation of fixed point equation equivalent to our equation in Corollary~\ref{cor: ueqn}. Here we use the figure as an illustration of \eqref{eq: decomp} in our proof of Theorem~\ref{thm: PGFHq} which we prove first before deriving the equations in Corollary~\ref{cor: ueqn}. 
The square represents an unknown tree cluster and the circle is the vertex connected to the randomly chosen initial edge below.
Figure was using MS Office.
}
\label{fig: H1}
\end{figure}


Evaluating any PGF at $x=1$ should yield $1$ since this is just the sum of the probabilities of a partition of event space.
By setting $H_q(1)=1$ in \eqref{eqn: HqEqn} in Theorem~\ref{thm: PGFHq}, we derive the following corollary.

\begin{corollary}
\label{cor: ueqn}
For a random graph $\mathbb{G}_{N,\textbf{d}}$ that satisfies the conditions of Theorem~\ref{thm: PGFHq}, the following must also be satisfied,
    \begin{equation}
    \label{eqn: u}
        u=G_q(u)
    \end{equation}
\end{corollary}
 
Corollary~\ref{cor: ueqn} provides a fixed point equation that we can use to solve for $u$. 
This is the same equation that is derived by~\cite{NewmanStrogatzWatts:2001}, however we believe the proof here is presented in a clearer manner using conditional probabilities that clearly outline when and how the PGF properties from Theorem~\ref{thm: pgf} apply.
For completion we also present the following elementary theorem which fully characterizes the solutions of this fixed point equation.

\begin{theorem}
\label{thm: usol}
    Consider the solution of equation \eqref{eqn: u} in $[0,1]$.
    \begin{enumerate}
        \item $u=1$ is always a solution to \eqref{eqn: u}.
        \item If $p_1=0$ and,
        \begin{enumerate}
            \item $p_2<1$ then $u=0$ is the only solution to \eqref{eqn: u} in $[0,1)$. 
            \item $p_2=1$ then \eqref{eqn: u} is satisfied for all $u\in[0,1]$.
        \end{enumerate}
        \item If $p_1>0$ and,
        \begin{enumerate}
            \item $G_q'(1)\leq 1$ then there is no solution to \eqref{eqn: u} in $[0,1)$.
            \item $G_q'(1)>1$ then there is a solution $u\in(0,1)$ to \eqref{eqn: u}.
        \end{enumerate}
    \end{enumerate}
\end{theorem}

\begin{proof}
By Theorem~\ref{thm: pgfGq}, we know that $G_q(1)=\frac{G_p'(1)}{G_p'(1)}=1$. Thus $u=1$ is always a trivial solution of equation \eqref{eqn: u}. This proves part 1.

We know that all the derivatives of the PGF $G_q$ must be non-negative for $x\in[0,1]$.
Recall also from the proof of Theorem~\ref{thm: pgfGq} that,
\begin{equation}
    G_q(x) =\frac{G_p'(x)}{G_p'(1)}=\frac{\sum_{k=1}^\infty kp_k x^{k-1}}{\langle K\rangle}.
\end{equation}
We also note that $G_q(0)=\frac{p_1}{\langle K\rangle}
=\frac{p_1}{p_1+2p_2+3p_3+\dotsc}\in[0,1]$.

Consider the case $p_1=0$.
Clearly $u=0$ is a solution to \eqref{eqn: u}. 
Suppose first that in addition we have $p_2=1$ then we must have $p_k=0$ for all $k\ne 2$ and $G_q(x)=x$ for all $x\in[0,1]$ which makes all points in that interval a solution of \eqref{eqn: u}. This proves part 2(b).
If instead of $p_2=1$ we assume $p_2<1$ then there exists a $k>2$ such that $p_k>0$ which means both $G_q'(x)$  and $G_q''(x)>0$ for $x\in(0,1]$. Since $G_q(x)$ is increasing and concave up on $(0,1]$. Since we already have $G_q(0)=0$, there is at most only one more solution to \eqref{eqn: u} and this is guaranteed by part 1 to be $u=1$. This proves part 2(a) and completes the proof of part 2.


Now assume that $p_1\in (0,1)$ so that $G_q(0)\in (0,1)$. 
If $p_2=1-p_1$ then for all $x\in[0,1]$, $G_q'(x)=\frac{2p_2}{p_1+2p_2}\in(0,1)$. Thus in this case,
$G_q(x)$ is a linear function with slope less than one so \eqref{eqn: u} can have at most one solution, and from part 1 this solution must be $u=1$. 
If instead we have $p_2<1-p_1$ then there exists a $k>2$ such that $p_k>0$ which means both $G_q'(x)$  and $G_q''(x)>0$ for $x\in(0,1]$.
Thus $G_q(x)$ is increasing and concave up on $(0,1]$. Since at $x=0$ we have $G_q(0)>0$, there can be up to two solutions to \eqref{eqn: u} on $(0,1]$. From part 1 we know that one solution is $u=1$. 
Using standard arguments, it is clear that there is another solution $u\in(0,1)$ to \eqref{eqn: u} if $G_q'(1)>1$ and none if $G_q'(1)\leq 1$. This completes the proof of part 3. We note that the first case $p_2=1-p_1$ is included in part 3(a).
\end{proof}

Using similar ideas to Theorem~\ref{thm: PGFHq}, we can find an alternative expression of $H_p(x)$ in terms of the solution $u$ in Corollary~\ref{cor: ueqn} and $H_q(x)$.
\begin{theorem}
\label{thm: PGFHp}
For a random graph $\mathbb{G}_{N, \textbf{d}}$ that satisfies the conditions of Theorem~\ref{thm: PGFHq}, the PGF $H_p(x)$ must satisfy the following equation:
    \begin{equation}
    \label{eqn: HpEqn}
        H_p(x)=\frac{x}{1-\mathcal{S}} G_p(u H_q(x))
    \end{equation}
\end{theorem}

\begin{proof}
Let $a$ be a randomly chosen vertex. 
The probability that $a$ belongs to a tree component is equal to $1-\mathcal{S}$, where $\mathcal{S}$ is the portion of vertices belongs to the giant component.
Using similar steps as in \eqref{eqn: CondProb}, we have the probability that $a$ has degree $k \in \mathbb{Z}_{>0}$ given that $a$ belongs to a small component:
\begin{align}
    \mathbb{P}&(\text{Degree of the chosen vertex }a\text{ is } k | \text{the vertex }a \in \text{ a tree component} )
    \nonumber
\\
    &= \frac{\mathbb{P}(D_G(a)=k \text{ and all stubs of }a \text{ pair to stubs with tree sub-components} ) }{ \mathbb{P}( a \in \text{ a tree component})}
    \nonumber
\\
    &=  \frac{p_k u^k}{1-\mathcal{S}}
\end{align}
Using similar steps as in Theorem~\ref{thm: PGFHq}, we note that the sum of the sizes of all the tree sub-components connected to the stubs of $a$ will have PGF given by,
\begin{equation}
    \label{eqn: Hp}
\sum_{k=1}^\infty \frac{p_ku^k}{1-\mathcal{S}}(H_q(x))^k.
\end{equation}
Since all these sub-components are trees, by Theorem~\ref{thm: component}, to get the size of the sub-component that $a$ belongs to we just need to add one (corresponding to $a$ itself) to the sum of the sizes of the sub-components, resulting in,
\begin{equation}
H_p(x)
=x\sum_{k=1}^\infty \frac{p_ku^k}{1-\mathcal{S}}(H_q(x))^k
=\frac{x}{1-\mathcal{S}}\sum_{k=1}^{\infty} p_k (u H_q(x))^k=\frac{x}{1-\mathcal{S}}G_p(u H_q(x))
    \end{equation}
\end{proof}

The previous theorem leads to the following result which is similar to the result we have in Corollary~\ref{cor: ueqn}.
\begin{corollary}
\label{cor: Seqn}
 For a random graph $\mathbb{G}_{N, \textbf{d}}$ that satisfies the conditions of Theorem~\ref{thm: PGFHq}, the following must also be satisfied,
    \begin{equation}
    \label{eqn: S}
        \mathcal{S}=1-G_p(u)
    \end{equation}
\end{corollary}
\begin{proof}
From $H_p(1)=1$ and $H_q(1)=1$ we derive $1=\frac{1}{1-s}G_p(uH_q(1))=\frac{G_p(u)}{1-\mathcal{S}}$. Rearranging yields the required result.
\end{proof}

Theorem~\ref{thm: PGFHq} and Theorem~\ref{thm: PGFHp} are refinements of the results of \cite{NewmanStrogatzWatts:2001} and \cite{Newman:2002}.
The final results are given in Corollary~\ref{cor: ueqn} and Corollary~\ref{cor: Seqn} which yield the same expressions as those in \cite{NewmanStrogatzWatts:2001} and \cite{Newman:2002}.
However in the original papers, the PGFs is defined based on $P(s)$ instead of $\hat{P}(s)$, and the equations were based on PGFs that were not explained to be unconditional probabilities.
This leads to confusion: $P(s)$ as a probability mass function is normalized only if $\mathcal{S}=0$, which occurs when all components are small components.
Thus the PGFs they defined would only be PGFs in that case. This makes it confusing to justify the resulting equations as equations for $u$ and $\mathcal{S}$ in the cases where the giant component exists.
Using the conditional probabilities in Definition~\ref{def: Hp} and \eqref{eqn: CondProb}, we obtained the same results for $u$ and $\mathcal{S}$ without confusion.

For a random graph, with large enough $N$ that has a small component, both \eqref{eqn: u} and \eqref{eqn: S} should be satisfied as a system at the same time.
By Theorem~\ref{thm: usol}, we have following observations:
\begin{enumerate}
\item We have $u=1$ if and only if $\mathcal{S}=0$. This pair of values is always a solution to \eqref{eqn: u} and \eqref{eqn: S} so it is always possible for the random graph $\mathbb{G}_{N,\textbf{d}}$ to have no giant component.

\item In the case when there is at most 1 solution of $u \in (0,1)$:
If such a $u$ exists then there is a corresponding $\mathcal{S} \in (0,1)$ that is a solution of the system.

\item For the case $p_1>0$, $G_q'(1)$ gives the threshold for existence of giant component.
If $G_q'(1) >1$, there is a positive probability that $\mathbb{G}_{N,\textbf{d}}$ has a unique giant component.
If $G_q'(1) <1$, there is no giant component in $\mathbb{G}_{N,\textbf{d}}$.
    
\item If a giant component exists, the portion of vertices in $\mathbb{G}_{N,\textbf{d}}$ that belongs to it is given by $\mathcal{S}=1-G_p(u) \in (0,1)$. 
\end{enumerate}

With $H_p(x)$ and $\mathcal{S}$, one can find the probability distribution for sizes of small components $\hat{P}(s)$ by differentiation of $H_p(x)$. 
However,  finding arbitrary derivatives of $H_p(x)$ in closed form is difficult \citep{Newman:2002}. The distribution can instead be evaluated more easily by numerical contour integration with the Cauchy formula.
\begin{equation}
    \hat{P}(s)=\frac{1}{s!}\frac{d^s H_p}{dx^s}|_{x=0}=\frac{1}{2 \pi i} \oint \frac{H_p(\zeta)}{\zeta^{s+1}} d \zeta
    \label{eqn: Contour Ps}
\end{equation}
By \eqref{eqn: hatPs}, one can find $P(s)$ directly from \eqref{eqn: Contour Ps} with $\mathcal{S}$.

The deterministic results we derived depends on the size of the network going to infinity.
In this case the random paring of stubs in Algorithm~\ref{alg: Config} can be seen as independent, since probability of pairing new stubs would not be affected by existing edges.
In practice we work with finite size networks, but it is not hard to see from the results of infinite size network, that these results are still good approximations for majority of the possible graphs in $\mathcal{G}_{N,\textbf{d}}$ for large $N$.
Using simulations, \cite{NewmanStrogatzWatts:2001} and \cite{Newman:2002} showed that the large component size given by $\mathcal{S}$ in this section is numerically the expectation giant component size among all possible graphs.
However, it remains an open question to prove analytically and rigorously that results of $\mathcal{S}$ is the mean-field result and exact among all random $\mathbb{G}_{N,\textbf{d}}$ for any degree distribution $d$.

As a partial proof, here we demonstrate that the threshold given by $G_q'(1)$ and the giant component size estimate $\mathcal{S}=1-G_p(u)$ we derived are equivalent to the results from Theorem~\ref{thm: Giant Component Size}, which show the exactness under some condition on $\textbf{d}$.
The second part of the next theorem is also mentioned in \cite{NewmanStrogatzWatts:2001}, but we provide details of proof here.

\begin{theorem}
For a random graph $\mathbb{G}_{N, \textbf{d}}$ that satisfies the conditions of Theorem~\ref{thm: Giant Component Size}, we have the following:
    \begin{enumerate}
        \item $\Lambda=\sum_{k=1}^\infty k (k-2) p_k=0$ if and only if $G'_q(1)=1$.

        \item $\mathcal{S}=1-G_p(u)=\Theta$ for $u \in (0,1)$ that is a solution of $u=G_q(u)$.
    \end{enumerate}
\end{theorem}

\begin{proof}
To prove part 1, by Theorem~\ref{thm: pgfGq} we have
    \begin{align}
        & G'_q(1)=\frac{G''_p(1)}{\langle K \rangle}=1
    \\
        \Leftrightarrow \quad & \sum_{k=1}^\infty p_k k = \langle K \rangle =G''_p(1)= \sum_{k=2}^\infty p_k k (k-1)
    \\
        \Leftrightarrow \quad & 0 = -p_1 + \sum_{k=2}^\infty p_k (k(k-1)-k)=\sum_{k=1}^\infty p_k k (k-2) = \Lambda
    \end{align}

To prove part 2, let $u(\alpha)=(1-\frac{2 \alpha}{\langle K \rangle})^{\frac{1}{2}}$. This is a bijection and we can solve $\alpha=\frac{1}{2} \langle K \rangle (1-u(\alpha)^2)$. We see that $u=1$ if and only if $\alpha=0$.
Recall the function $f(\alpha)$ from Theorem~\ref{thm: Giant Component Size}.
\begin{align}
    f(\alpha) &=\langle K \rangle-2 \alpha- \sum_{k=1}^{\Delta} k p_k \Big(1-\frac{2 \alpha}{\langle K \rangle}\Big)^{\frac{k}{2}}
    \nonumber
    \\
    &=\langle K \rangle- \langle K \rangle (1-u(\alpha)^2) - \sum_{k=1}^{\Delta} k p_k \Big(1-\frac{ \langle K \rangle (1-u(\alpha)^2)}{\langle K \rangle}\Big)^{\frac{k}{2}}
    \nonumber
    \\
    &= \langle K \rangle u(\alpha)^2 - \sum_{k=1}^{\Delta} k p_k  u(\alpha)^k
\end{align}
Thus, for $u \neq 0$, we have:
\begin{align}
    &  f(\alpha) = 0\
    \\
    \Leftrightarrow \quad & 0= \langle K \rangle u(\alpha)^2 - \sum_{k=1}^{\Delta} k p_k u(\alpha)^k = \langle K \rangle u(\alpha)^2 - u(\alpha) \sum_{k=1}^{\Delta} p_k k u(\alpha)^{k-1}
    \\
    \Leftrightarrow \quad & 0 = u(\alpha) (\langle K \rangle u(\alpha) - G'_p(u(\alpha)))
    \\
    \Leftrightarrow  \quad & u(\alpha) = \frac{G'_p(u(\alpha))}{\langle K \rangle}= G_q(u(\alpha))
\end{align}
So $f(\psi)=0$ if and only if $u(\psi)=G_q(u(\psi))$.
It follows that,
\begin{equation}
\Theta=1-\sum_{k=1}^{\Delta} p_k \Big(1-\frac{2 \psi}{\langle K \rangle}\Big)^{\frac{k}{2}}=1- \sum_{k=1}^{\Delta} p_k u(\psi)^k= 1- G_p(u(\psi))= \mathcal{S}
\end{equation}
This completes the proof of this theorem.
\end{proof}

This result shows that under the conditions of Theorem~\ref{thm: Giant Component Size}, $\mathcal{S}$ provides a good estimate of the giant component size for a random graph $\mathbb{G}_{N, \textbf{d}}$, and $G'_q(1)$ is indeed a threshold for giant component to exist.

\section{Disease transmission on networks}
\label{sec: disease}

In this section we extend the ideas from random graph theory to disease transmission on networks.
The networks we will consider are graphs wherein the vertices represent individuals in the population and edges represent a disease-causing connection between the two vertices it connects to.
With that in mind, all the model assumptions and terminologies in this section will be listed in terms of network (graph) definitions, and one can easily translate them back to its epidemiological meaning. 
We make the following simplifying assumptions about our disease transmission network for any fixed time:
\begin{enumerate}
    \item[(A1)] The disease can be transmitted along an edge from one vertex to another in both directions.
    
    \item[(A2)] There is at most one edge between any pair of vertices.

    \item[(A3)] No edge starts and ends at the same vertex.

    \item[(A4)] There are infinitely many vertices in the network.
\end{enumerate}

A simple, undirected graph with infinite size, like the graphs we discussed in Section~\ref{sec: basic}, can be used to model a disease network that satisfies (A1)--(A4).
While in practice the network has to be finite, if it is large enough, we can observe results approximated by the results we obtain for $N\to\infty$ from Section~\ref{sec: randomgraphs}.

As we previously discussed, an exact graph of any network is usually hard to determine and thus we represent them instead by random graphs generated and characterized by some given degree sequence. 
Thus we will call a random graph a network for convenience, so that all definitions for graphs and random graphs in Section~\ref{sec: review} now also apply to the network.

To begin, we only consider a network that does not change over time by making the following additional assumption on the disease network:
\begin{enumerate}
    \item[(A5)] The edges between vertices do not change over time.
\end{enumerate}
Since such random graphs are generated using the configuration model like in Algorithm~\ref{alg: Config}, in network epidemiology terminology we refer to such a network following (A1)-(A5) also as a configuration model (CM) network.
We will focus on these types of networks for the remainder of this paper except for Section~\ref{sec: MFSH} where we introduce the mean-field social heterogeneity model.

Configuration networks, or more general network models can be used with many types of epidemiological models, but we will focus on its interaction with the SIR compartmental model, which is how network epidemiology methods were originally developed and mostly used. 
We refer to this as a network-SIR model and employ the assumptions from the SIR compartmental model, so they would share common assumptions and parameters as much as possible.
Thus, other than (A1)--(A5), we have several more assumptions to make. The first one splits the vertices into different compartments according to disease status.
\begin{enumerate}
 \item [(A6)] At all possible times $t$, all vertices in the network can be classified into only one of the tree compartments: susceptible, infectious and recovered.
Let $S(t)$, $I(t)$ and $R(t)$ to be the proportion of the vertices in the network belongs to susceptible, infectious and recovered compartment at time $t$ and these functions are found in the simplex given by,
    \begin{equation}
        (S(t),I(t),R(t))\in \mathbb{S} := \{(S,I,R)|S, I, R \in[0,1], S+I+R=1\}, \quad \; \forall t\geq 0.
    \label{eqn: simplex}
    \end{equation}
\end{enumerate} 

The MA-SIR model we discussed in the introduction is usually applied to a large enough population so that $S(t)$, $I(t)$ and $R(t) \in [0,1] \subset \mathbb{R}$ can be approximated by continuous functions defined on a continuous time interval.
Using such a setting, one can derive ODE systems to describe the disease spreading and use techniques from ODE theory.
This is a common approach applied by many continuous models with large discrete quantities so that more mathematical tools from continuous mathematics and differential equations can be used.
Using the same large population assumption in (A4), we also apply same continuous dimensionless approach in (A6), even for typical percolation process, ODE and dynamical information are not required.
In network models, these proportions can also be interpreted as the probability that a randomly chosen individuals belongs to each compartment at time $t$.
Later on in Section~\ref{sec: dyn}, we reconsider the dynamics with ODEs in network models, so such continuity in quantities is required and we just keep it as our default assumption for consistency.

In the SIR model, the disease transmission is represented by flows between compartments.
Similarly in network-SIR model, we assume:
\begin{enumerate}
    \item [(A7)]Only susceptible vertices can get infected and change to infectious vertices, but never change into recovered vertices directly.
    Moreover, any susceptible vertex $j$ can only get infected from an edge $e$ connecting itself to an infectious vertex $i$ with a probability $\beta_{ij}$.
    
    \item [(A8)]Any infectious vertex $i$ remains infectious for a period of time $\tau_i$, then changes to a recovered vertex.
    
    \item [(A9)]Once a vertex becomes recovered, it would remain in that compartment forever and never change back to susceptible or infectious. 
\end{enumerate} 

In extensions of the SIR compartmental models, there are multiple different detailed assumptions about how the probability $\beta_{ij}$ in (A7) and function of recovery time $\tau_i$ in (A8) is defined, which would also lead to slightly different models.
We will discuss some of the detailed assumptions for these functions in Section~\ref{sec: Trans}.

The assumption (A9) is also known as the assumption of perfect immunity.
This protection remains during the single epidemic outbreak.
\cite{BansalMeyers:2012} has extend the model and percolation method to the case considering imperfect immunity between multiple outbreaks, however this is beyond the scope of our current study.

Both MA-SIR models and network-SIR models need assumptions for the initial state of the disease.  
In MA-SIR models, such assumptions relate to the initial condition in order to set up the ODE system as an initial value problem. To model a new epidemic outbreak, we usually assume that we initialize from only one infectious individual at time $t=0$, while all others in the population are susceptible and there is no recovered individual.
In MA-SIR models, this becomes the initial condition: $S(0)=1-I(0), I(0)>0, R(0)=0$, where $I(0)$ is assumed to be some small amount close to $0$.
If there is a specified population size $N \in \mathbb{Z}_{>0}$, we can assume that $I(0)=\frac{1}{N}$.

For the network-SIR model, a discrete perspective of the graph needs to 
be considered to apply theories in Section~\ref{sec: basic}: the numbers of vertices, edges and components in networks must have positive integer values.
Therefore, we have the following assumption for initial stage of disease:
\begin{enumerate}
    \item [(A10)] At time $t=0$, exactly one uniformly chosen vertex is infectious and all other vertices are susceptible.
\end{enumerate} 

In the MA-SIR compartmental models we can compute then next-generation matrix method to derive the basic reproduction number ($R_0$) to determine the stability of the disease-free equilibrium~\citep{diekmann_heesterbeek_metz_1990,vandendriessche_watmough_2002}. 
In most cases with small initial infection amount, when $R_0<1$ the disease will only infect a small portion of the population and asymptotically go extinct. 
If $R_0>1$ there is a significant outbreak.
This transcritical bifurcation at $R_0=1$ is similar to the critical transition threshold from the network consisting only of small components to the emergence of the giant component as we have discussed in Section~\ref{sec: pgfs}.

We will derive threshold results for the network-SIR model using percolation methods which is a common approach to analytically connect networks with compartmental dynamics.
We will derive threshold results for the network-SIR model using percolation methods which is a common approach to analytically connect networks with compartmental dynamics.
Percolation process on a network is that some of its vertices or edges are randomly assigned ``occupied'' status.
Then one can explores various properties of the resulting patterns of occupation status.
The percolation model was first developed in the 1950s and one of its main application was the modeling of the disease spreading on networks.
As showing in Figure~\ref{fig: percolationtypes}, there are two types of percolation methods: site percolation and bond percolation.

\begin{figure}[hbt]
    \begin{center}
        \includegraphics[width=0.6\textwidth]{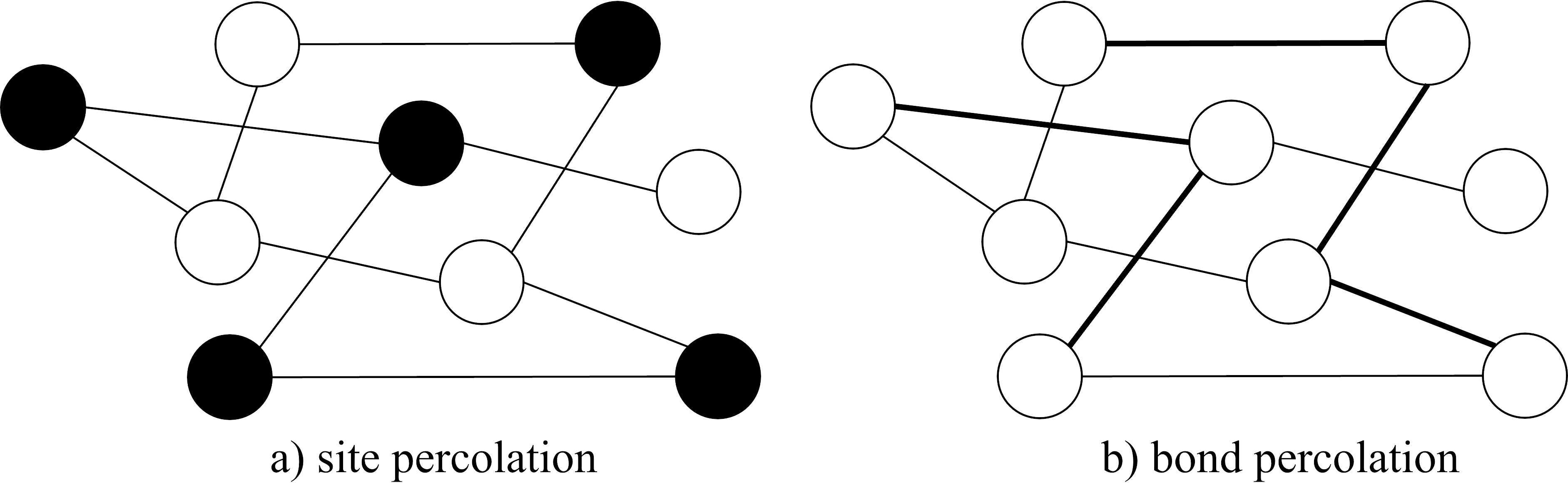}
    \end{center}
    \caption[Site and bond percolation on a network]{Site percolation a) focus on whether vertices (sites) are occupied (solid circles) or not (open circles). Bond percolation b) focus on whether the edges (bonds) are occupied (thick lines) or not (thin lines).
    Figure was created using MS Office.}
\label{fig: percolationtypes}
\end{figure}

Under assumptions (A1)-(A4), since edges and vertices are random with no specified structure, the occupation status of both edges and vertices are random events under some probability depends on the degree distribution.
Not every disease-causing connection, or edge, in the network will transmit the infection.
The typical bond percolation process developed by \cite{Newman:2002} only tracks the overall probability of a randomly chosen edges actually transmitting infection if it is joined with an infected vertex.
In Section~\ref{sec: Trans}, we will see that with further assumptions, such probability can be determined as a function of disease parameters.

We use percolation process by defining the occupation status of vertices and edges in bond percolation process in the following manner:
\begin{definition}
In the bond percolation process of the network-SIR model,
\label{def: occupation}
    \begin{enumerate}
        \item An edge is considered to be \textbf{occupied} if it can transmit the infection when any of its ends is infected.
        \item A vertex is considered to be \textbf{occupied} if it is incident with an occupied edge. 
    \end{enumerate}
\end{definition}
Our process focuses on the occupation pattern of edges, and the proportion of occupied edges among all edges is the same as the overall prior probability.
This does not provide information on the time evolution of the epidemic as it spreads through the network.

A common misunderstanding when using bond percolation processes to model disease transmission is that occupied vertices is the same as infected vertices, and occupied edges are the edges that actually transmit the infections during the outbreak.
However, if both ends of an edge are infected, then the infection cannot transmit between them through the edge.
So there is no loop if we just consider infected vertices and edges that actually transmit the infection.
On the other hand, occupied vertices and occupied edges can form loops since no evolutionary information is considered when defining occupational status.
This difference is important since the existence of loops is equivalent to existence of a giant component w.h.p. as we explained in Section~\ref{sec: basic}.
We will discuss details on how infection and occupation are related in Section~\ref{sec: perco}, so that the result of percolation can be connected with epidemic results.
This also explains the necessity of (A10) for the network-SIR model.

In Section~\ref{sec: Trans} we present a mathematical description of the prior probability for disease spreading within network based on assumption (A1)-(A10) and some further detailed assumptions added on to (A7) and (A8) that would affect the spreading procedure. 
In Section~\ref{sec: perco} we derive methods to derive the expectation of the epidemic size on the network over all possible networks with the given degree distribution, based on (A1)-(A10) and Section~\ref{sec: Trans}.

\subsection{Transmissibility} 
\label{sec: Trans}




In the network-SIR model, the existence of an edge between two vertices does not mean the disease will be transmitted even if one end is infected.
Newman's percolation method, like many other disease models, assumes the probability that a disease is transmitted between a connected pair of infected and susceptible individuals is identically distributed for mathematical simplicity.
The transmissibility of the disease is the key idea behind this.

\begin{definition}[Transmissibility]
The \textbf{transmissibility} $T$ of a disease, also known as uniform edge occupation probability, is the mean (across all edges) probability that a randomly chosen edge is capable of transmitting the disease per unit time.
\label{def: Transmissibility}
\end{definition}

In order to determine the transmissibility $T$, we need to first determine the probability of transmission $T_{ij}$ between from an infected vertex $i$ to a susceptible vertex $j$. This relies on both $\beta_{ij}$ from (A7) and the recovery rate $\gamma_i$ from (A8).


\begin{theorem}[\cite{Newman:2002}]
\label{thm: not trans rate}
Assume (A1)--(A10). Let $T_{i j}$ be the probability of the infection transmitting from infected vertex $i$ to the susceptible vertex $j$ along edge $e$ per unit time.
Then,
    \begin{equation}
        T_{ij}=1-e^{-\beta_{ij} \tau_{i}}.
    \label{eqn: Tij}
    \end{equation}
\end{theorem}
\begin{proof}
Any randomly chosen edge $e$ has infectious probability equal to $\beta_{i j}$ and the infectious individual $i$ remains infectious for a time $\tau_i$.
Over a small time period $\Delta t$, the probability of transmission is $\beta_{ij} \Delta t$. 
Since the duration of the infected individual remains infected is $\tau_i$ during which time we have $\frac{\tau_i}{\Delta t}$ intervals of length $\Delta t$.
The probability of no transmission is given by,
\begin{equation}
        1-T_{i j} = \lim_{\Delta t\to 0} (1- \beta_{i j} \Delta t)^{\frac{\tau_i}{\Delta t}}=e^{-\beta_{i j} \tau_i}.
          \label{eqn: probtransmit}
\end{equation}
Rearranging yields the result of the theorem.
\end{proof}

To find an expression for $T=\langle T_{ij} \rangle$ (where the average is taken over all pairs $i$ and $j$), we need further assumptions on $\beta_{ij}$ and $\tau_i$.
The different forms that these assumptions can take is the focus of this section. 

In the MA-SIR model, the mass action assumption averages out the individual differences and assuming homogeneity not only in individual contact rates but also in disease transmission process. 
From the network point of view, the homogeneity in contact rate means each vertex has the same degree $N-1$, where $N$ is the network size.
The homogeneity assumption can be written as an additional assumption to (A7) and we denote this by (A7a).
\begin{enumerate}
\item [(A7a)] For any possible edge $e$ connecting susceptible vertex $j$ and infectious vertex $i$, the probability of the infection transmitting has the same uniformly constant value $\beta_{ij} = \hat{\beta} \geq 0$.
\end{enumerate}
Here $\hat{\beta}$ is the constant per-infected transmission rate in Figure~\ref{fig: MA-SIR}.

To demonstrate the dynamics of disease spreading as an ODE system, the MA-SIR model assumes that the recovery time of each infectious individual is independent and identically distributed (i.i.d.) exponential random variables with fixed recovery rate parameter.
This parameter, the constant per-capita recovery rate, is denoted by $\hat{\gamma}$ in Figure~\ref{fig: MA-SIR} for the ODE and the expectation of recovery time among all infectious individual is also required to be $\frac{1}{\hat{\gamma}}$.
This assumption can be written as an additional assumption to (A8) and we denote this by (A8a).
\begin{enumerate}
    \item [(A8a)] For any infectious vertex $i$, the recovery time $\tau_i$ is an i.i.d. exponential random variable (among all infectious vertices) with probability density function given by,
    \begin{equation}
    \label{eqn: exp tau}
        f_{\tau}(t)=\hat{\gamma}e^{-\hat{\gamma} t}   
    \end{equation}
\end{enumerate}
Here $\hat{\gamma}>0$ is uniform constant recovery rate among all infectious vertices in Figure~\ref{fig: MA-SIR}.
This density function results in  $\mathbb{E}_{ij} [\tau_{ij}]=\frac{1}{\hat{\gamma}}$.

The network models allow for social heterogeneity by removing the mass action assumption, so that transmissibility for the edges might vary if we have different assumptions about the disease spreading process.
The easiest way to include the heterogeneity is to consider different contact rates, which is now represented by degree sequences in the network, and keep assumptions (A7a) and (A8a) about disease transmission procedure. 
If so, the mathematical simplicity is kept mostly, therefore these assumptions are used by many of the compartmental network models, including those developed by \cite{MillerSlimVolz:2012} in Section~\ref{sec: dyn}.

To understand the transmissibility idea and the capabilities of the network-SIR model better, we can also derive $T$ in a more general case.
By allowing weaker assumption than (A7a) and (A8a), the typical bond percolation method could apply with higher level heterogeneity in disease transmission while maintaining reasonable tractability.
So instead of using identical constant parameters $\hat{\beta}$ and $\hat{\gamma}$ for all individuals, we can make alternative new assumptions (A7b) and (A8b) that respectively provide more restrictions than (A7) and (A8):
\begin{enumerate}
\item [(A7b)] For any possible edge $e$ connecting susceptible vertex $j$ and infectious vertex $i$, the probability $\beta_{i j}$ of infection transmitting  is an i.i.d. random variable among all such edges, chosen from a distribution with arbitrary probability density function $f_{\beta}(b)$.

\item [(A8b)] For any infectious vertex $i$, the recovery time $\tau_i$ is an i.i.d. random variable among all infectious vertices, chosen from a distribution with arbitrary probability density function $f_{\tau}(t)$.
\end{enumerate}


Since $\beta_{ij}$ and $\tau_{i}$ i are i.i.d. random variables, so is $T_{ij}$.
Therefore, the disease transmission probability between any infected-susceptible pair of individuals is simply the expectation of $T_{ij}$ with probability density function $f_{\tau}(t)$ and $f_{\beta}(b)$.

\begin{theorem} [\cite{Newman:2002}]
\label{thm: genT}
Under assumptions (A1)--(A10), (A7b) and (A8b), the transmissibility $T$ defined in Definition~\ref{def: Transmissibility} is given by,
    \begin{equation} 
        T=\mathbb{E}_{ij} [T_{ij}] =1-\int_{0}^{\infty} \int_{0}^{\infty} f_{\beta}(b) f_{\tau}(t) e^{-bt} db\, dt.
    \label{eqn: prob gen T}
    \end{equation}        
Clearly $T \in [0,1)$.        
\end{theorem}

Individual differences with this assumption has no impact on final result: from whole population level, the disease propagate as if all transmission always has same probabilities as overall probability $T$~\citep{Newman:2002}.
The typical bond percolation process requires (A7b) and (A8b) to get a uniformly constant $T$ to make the models solvable. 
If $\beta_{ij}$ and $\gamma_i$ are not i.i.d. then we would have distinguishable types of edges in the network, and thus different type of vertices and corresponding vary transmissibility.
This is beyond the scope of this paper, but we note that the idea behind the percolation process still works if we have a limited number of edge types, but this will need some modifications of the percolation process~\citep{BansalMeyers:2012}.


In order to compare the use of the percolation method on the network-SIR model with the MA-SIR model and other models to be discussed in Section~\ref{sec: dyn}, we use (A7a) and (A8a) as the default assumption.
By doing so, these models share the same assumption on disease transmitting procedure, thus the same parameters.
To derive $T$ under (A7a) and (A8a), we observe that (A8a) is a special case of (A8b) with $f_{\tau}(t)=\hat{\gamma}e^{-\hat{\gamma} t}$, and (A7a) is a special sub-case of (A7b) using a distribution that puts all the probability entirely at $\beta=\hat{\beta}$. 


\begin{corollary} 
\label{cor: const T}
Under assumptions (A1)--(A10), (A7a) and (A8a), the transmissibility $T$ defined in Definition~\ref{def: Transmissibility} is given by,
\begin{equation} 
        T=\frac{\hat{\beta}}{\hat{\beta}+\hat{\gamma}}
    \label{eqn: const T}
    \end{equation}  
Clearly $T \in [0,1)$.        
\end{corollary}
\begin{proof}
Using \eqref{eqn: exp tau} and $f_{\beta}(b)=\delta_{\hat{\beta}}({b})$ (the Dirac delta function with mass located at $\hat\beta$) in Theorem~\ref{thm: genT}, 
\begin{align}
    T & =\mathbb{E}_{ij} [T_{ij}] =1-\int_{0}^{\infty}\int_{0}^{\infty} \delta_{\hat{\beta}}(b)\hat{\gamma}e^{-\hat{\gamma} t}e^{-{\beta} t}  db\, dt
    \nonumber
    \\
    &=1-\int_{0}^{\infty} e^{-\hat{\beta} t} \hat{\gamma}e^{-\hat{\gamma} t} dt =1-\hat{\gamma}\int_{0}^{\infty} e^{-(\hat{\beta}+\hat{\gamma})\tau} d\tau=\frac{\hat{\beta}}{\hat{\beta}+\hat{\gamma}}
    \end{align} 
\end{proof}

Although we have already explained that (A7a) and (A8a) will be our default assumptions, we take the opportunity of this section to clarify the consequences of other assumptions that have been a cause for confusion in some of the early work on network epidemiology.
Between \cite{Newman:2002} and \cite{Newman:2003,Newman:2010} there is an assumption change for recovery time, which is only briefly discussed in \cite{Newman:2010}.
In the original derivation of transmissibility~\citep{Newman:2002} like in Corollary~\ref{thm: genT}, $T$ is derived from the MA-SIR compartmental model by replacing homogeneous assumption (A7a)-(A8a) with heterogeneous assumption (A7b)-(A8b).
However, this result does not include a clear assumption to specify the distribution of $\beta_{ij}$ and $\tau_i$, and there is no stated relationship between $\tau_i$ and the recovery rate $\gamma$ in the SIR model.
In \cite{Newman:2003}, the same assumption (A7b) is used but the other assumption changes from (A8b) to the next assumption (A8c) which we discuss below. 
\begin{enumerate}
    \item [(A8c)] For any infectious vertex $i$, the recovery time $\tau_i$ is given by 
    \begin{equation} 
        \tau_i= \frac{1}{\gamma_i}
    \label{eqn: taugamma A8c}
    \end{equation} 
    where $\gamma_i$ is the recovery rate of each vertex $i$. 
    Furthermore, $\gamma_i$ is i.i.d. random variable among all infectious vertices, chosen from some appropriate distributions with probability density function $f_{\gamma}(r)$.
\end{enumerate}
Under assumption (A8c), we rewrite the result of Theorem~\ref{thm: not trans rate} to,
\begin{equation}
    T_{ij}=1-e^{-\frac{\beta_{ij}}{\gamma_i}}
\label{eqn: Tij newman}
\end{equation}
Following the same ideas from Theorem~\ref{thm: genT}, we can derive the following result consistent with \cite{Newman:2003}:
\begin{theorem} [\cite{Newman:2003}]
\label{thm: genT newman}
Under assumptions (A1)-(A10), (A7b) and (A8c), the transmissibility $T$ defined in  Definition~\ref{def: Transmissibility} is given by,
\begin{equation} 
T=\mathbb{E}_{ij} [T_{ij}] = 1-\int_{0}^{+\infty} f_{\beta}(b) f_{\gamma}(r) e^{-\frac{b}{r}} db \, dr.
\label{eqn: prob gen T newman}
\end{equation}        
Clearly $T \in [0,1)$.        
\end{theorem}

Other than the heterogeneity on $\gamma_i$, (A8c) has another key difference about the relationship between recovery time and the recovery rate, comparing to the common assumption (A8a) for most compartmental model.
To better illustrate the difference and its impact, we again restrict (A8c) to have no individual difference on $\gamma_i$:
\begin{enumerate}
    \item [(A8d)] For any infectious vertex $i$, the recovery time $\tau_i$ is given by, 
    \begin{equation} 
        \tau_i= \frac{1}{\hat{\gamma}}
    \label{eqn: taugamma A8d}
    \end{equation} 
\end{enumerate}
Again, $\hat{\gamma}$ is uniform constant recovery rate among all infectious vertices in Figure~\ref{fig: MA-SIR}, so now (A8d) has the same parameter as (A8a).
Moreover, they both have the same expected recovery time $\mathbb{E}_{i,j}[\tau]=\frac{1}{\hat{\gamma}}$, which provide a better comparison.

\begin{figure}[htbp]
    \begin{center}
        \includegraphics[width=0.5\textwidth]{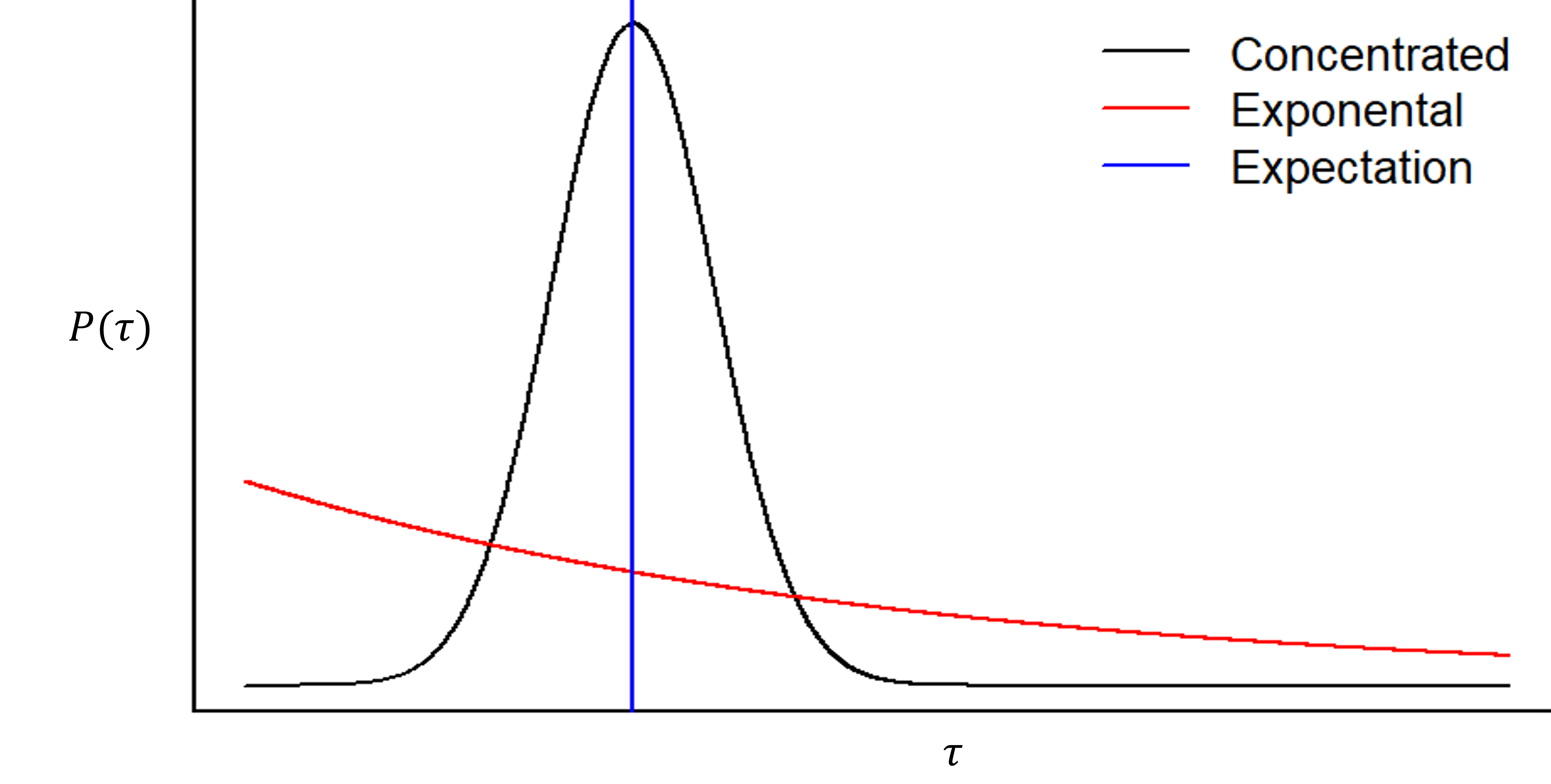}
    \end{center}
    \caption{Probability density function $P(\tau)$ with respect to recovery time $\tau$ with the same expectation $\mathbb{E}_{i,j}[\tau]=\frac{1}{\hat{\gamma}}$. The black is curve is the concentrated distribution in \cite{Newman:2010}. The red curve is the exponential distribution in (A8a). The blue line is the constant recovery time $\tau=\frac{1}{\hat{\gamma}}$ in (A8d).
    Figure was created using R}
\label{fig: recovery time}
\end{figure}

\cite{Newman:2010} explains the motivation and justification of (A8d), which extend to (A8c) with heterogeneity.
Unlike in (A8a) where individuals are assumed to recover after an exponentially distributed time, mainly for simplicity in ODEs, it is claimed that typically diseases have more concentrated recovery times, with a density function more similar to the black curve in Figure~\ref{fig: recovery time}.
Under such circumstance, (A8d) serves as a better simplification of the concentrated distribution than exponential distribution in (A8a).
Then based on Theorem~\ref{thm: genT newman} for (A8c), we have the following result: 
\begin{corollary} 
\label{cor: const T Newman}
Under assumptions (A1)--(A10), (A7a) and (A8d), the transmissibility $T$ defined in Definition~\ref{def: Transmissibility} is given by,
    \begin{equation} 
        T= T_{ij} = \mathbb{E}_{ij}[T_{ij}]= 1-e^{-\frac{\hat{\beta}}{\hat{\gamma}}}
    \label{eqn: const T Newman}
    \end{equation}  
Clearly $T \in [0,1)$.  
\end{corollary}

\begin{figure}[t]
    \begin{center}
        \includegraphics[width=1\textwidth]{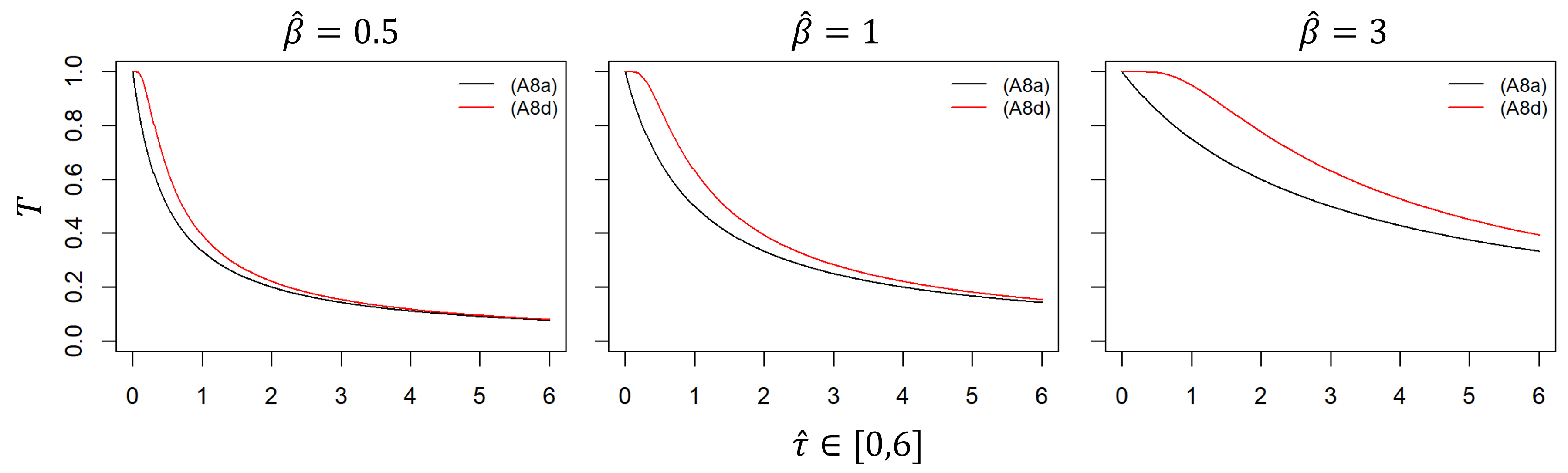}
    \end{center}
    \caption{Transmissibility $T$ for $\hat{\beta}=0.5$, $1$, $3$ and $\hat{\tau}\in (0,6)$ under different assumption.
    Figures were created using our R Package}
\label{fig: T compare}
\end{figure}

The transmissibility $T$ in Corollary~\ref{cor: const T} and Corollary~\ref{cor: const T Newman} have mostly the same assumptions and parameters other than the difference in recovery time illustrated in Figure~\ref{fig: recovery time}.
We can compare these expressions under different value of parameters in Figure~\ref{fig: T compare}.
Here the black curve is the case with (A8a) and Corollary~\ref{cor: const T}, while the red curve is the case with (A8d) and Corollary~\ref{cor: const T Newman}.
We see from the graphs and the series expansion of \eqref{eqn: const T Newman} that (A8d) leads to higher transmissibility than (A8a), even with same parameter $\hat{\beta}$ and $\hat{\gamma}$.

As we discussed earlier, moving forward we always use assumptions (A7a)--(A8a), in conjunction with (A1)--(A10), to determine the transmissibility $T$ unless otherwise specified.
Under this scenario, $T$ is given by Corollary~\ref{cor: const T}.
However, it is also worth to mention that, all general result of the typical bond percolation model presented in next Section~\ref{sec: perco} is presented in terms of $T$.
Since the more detailed assumptions that we discussed in this section for (A7) and (A8) would only affect $T$, such general results are independent with these detailed assumptions and just require (A1)--(A10).

\subsection{Percolation Process and Result}
\label{sec: perco}

Having now defined the transmissibility $T$, we can continue applying bond percolation to find the occupation status of edges and vertices.
Important symbols from previous sections that are relevant to the next sections are summarized in Table~\ref{table: parameters} along with their interpretations for modeling disease networks and mathematical definitions.

\begin{center}
    \begin{table}[htbp]
    \caption{Important notation referring to the network}
    \label{table: parameters}
        \begin{tabular}{|C{0.5in}|L{1.8in}|L{1.8in}|L{1.2in}|}
            \hline
            Symbol & Interpretation & Mathematical Definition & Range of values
            \\
            \hline
            $K$ & Number of connections of a randomly selected individual & Random variable denoting the degree of a randomly chosen vertex & Each evaluation has a value in $\mathbb{Z}_{\geq 0}$
            \\
            \hline
            $p_k$ & Proportion of individuals with degree $k$ & Probability mass function of $K$ &  $p_k\in[0,1]$, $\sum_{k=0}^\infty p_k=1$
            \\
            \hline
            $G_p(x)$ &  \multicolumn{2}{c|}{Probability generating function of $K$} & Refer to \eqref{eqn: pgfgp}
            \\
            \hline
            $G_q(x)$ &  \multicolumn{2}{c|}{Probability generating function of excess degree} & Refer to \eqref{eqn: pgfGq}
            \\
            \hline
            $S(t)$ & Proportion of the population that is susceptible at time $t$. & Probability that a randomly selected vertex is susceptible. & Refer to \eqref{eqn: simplex}
            \\
            \hline
            $I(t)$ & Proportion of the population that is infectious at time $t$. & Probability that a randomly selected vertex is infectious. &  Refer to \eqref{eqn: simplex}
            \\
            \hline
            $R(t)$ & Proportion of the population that is recovered at time $t$. & Probability that a randomly selected vertex is recovered. &  Refer to \eqref{eqn: simplex}
            \\
            \hline
            $\hat \beta$ & \multicolumn{2}{c|} {Uniform per-infected transmission rate}  & $[0,\infty)$
            \\
            \hline
            $\hat \gamma$ & \multicolumn{2}{c|}{Uniform per-infected recovery rate}  & $(0,\infty)$
            \\
            \hline
            $T$ & {Transmissibility refer to Definition~\ref{def: Transmissibility}} & Refer to Corollary~\ref{cor: const T} & $[0,1)$
            \\
            \hline
        \end{tabular}
    \end{table}
\end{center}

As in Definition~\ref{def: occupation}, for a randomly chosen edge $T$ is the probability it is occupied, and all vertices that connect to occupied edges are occupied as well.
Since within existing components, not every edge and vertex are occupied, a small initial number of infected vertices will not necessarily infect the whole component.
However, if the spreading initiates from one occupied vertex, then all the occupied vertices and the occupied edges connected to the initial infected occupied vertex will form a connected subgraph.
The structure of this subgraph will be determined by the underlying component that the initial vertex belongs to.
If we only focus on occupied edges and discard those that are not, we derive a new network based on the previous one, where the maximum occupied connected subgraph is now a component for the new graph.
\begin{definition}[Occupied Component]
    In a network with occupation defined in Definition~\ref{def: occupation}, 
    \begin{itemize}
        \item An \textbf{occupied component} is the maximal subset of occupied vertices that are connected by occupied edges through the network.
        \item The \textbf{occupied component size} is the number of occupied vertices belongs to such components.
    \end{itemize}  
    \label{def: Occupied Component}
\end{definition}

\begin{figure}[htbp]
    \begin{center}
        \includegraphics[width=1\textwidth]{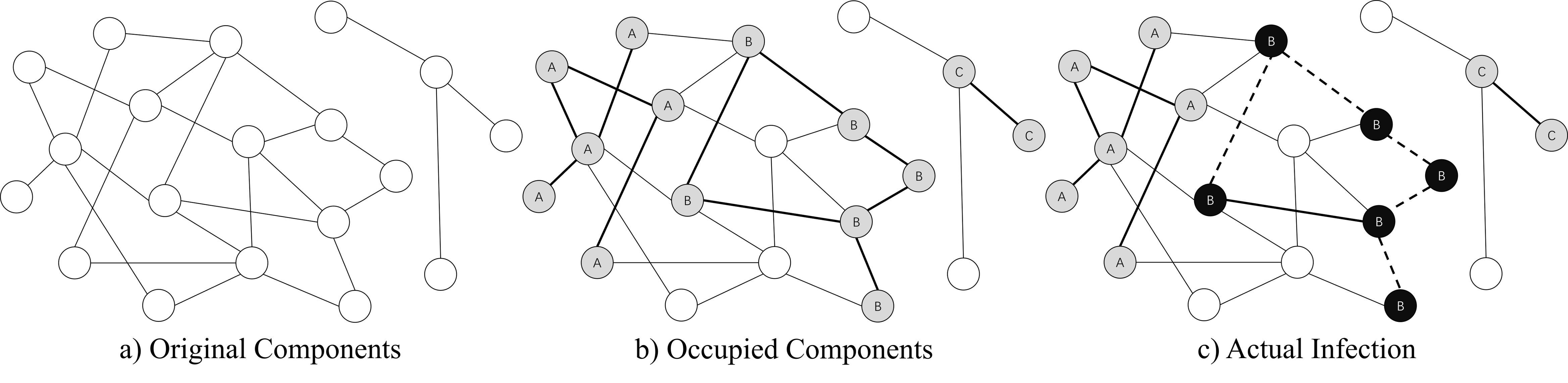}
    \end{center}
    \caption{Relationship between original components and occupied components. a) Two components from original network before assigning occupation. b) Three occupied components after assignment. Gray vertices and thickened solid lines are occupied vertices and edges accordingly. Each vertex in the occupied components is labeled by letters A, B, C. c) The actual infection initialized from one of the occupied vertices in occupied component B. Black vertices and dashed thickened lines are respectively infected vertices and edges that transmitted the infection.
    Figures were created using MS Office}
\label{fig: occupied components}
\end{figure}

We show a simple example on a small specified network in Figure~\ref{fig: occupied components} to illustrate the relationship between occupied component and the underlying original component.
Here are several important observations for the occupied components:
\begin{enumerate}
    \item [(a)] Since the occupation of edges is based on original edges, no new edges are created and the occupied component is a subgraph on original component.
    More precisely, size of occupied component is always no larger than the components it belongs to, so only occupied components that belong to a giant original component can be a giant occupied component.
    
    \item [(b)] The infection will never spread across different occupied components. 
    The infection cannot be transmitted between different occupied components even if they are connected by an unoccupied edge like A and B, since by definition, only occupied edges can actually transmit the disease.
    
    \item [(c)] Not every occupied vertex is actually infected eventually and not every occupied edge will transmit the infection.
    If the initial infectious vertex does not belong to the occupied component, like A and C, then none of its occupied vertices will be infected and none of its occupied edges will transmit the disease.
    
    \item [(d)] If any of the occupied vertices in an occupied component is the initial infectious vertex, like in B, then every occupied vertex in this occupied component will be eventually infected.
    However, if there are loops then there is always an occupied edge in each loop that does not actually transmit the disease. 
\end{enumerate}

This example is just for understanding. In the percolation process for large random networks, there is no structure specified thus all the occupation status and infection status is randomly assigned with some probability.
In the next results we denote vertices that have been infected to be infected vertices. We note that this includes vertices representing individuals that are (currently) infectious and those that have recovered.

\begin{theorem}
\label{thm: infection occupation}
    Under assumption (A1)-(A10), all infectious vertices belong to the same single occupied component that contains the initial single infected vertex.
    Moreover, all occupied vertices in such occupied component are infected vertices. 
\end{theorem}
\begin{proof}
This follows from observations (c)--(d) and assumption (A10) that exactly only one vertex is initially infectious.
\end{proof}

We call such occupied component the infected component, and we have further corollary as directly result from Theorem~\ref{thm: infection occupation}:
\begin{corollary}
\label{cor: infection occupation}
    Under assumption (A1)-(A10), we have the following statements are true:
    \begin{enumerate}
        \item The total number of infected vertices after the outbreak is the same as the occupied component size of the infected component.
        \item The probability that an occupied component is the infected component is the same probability that a randomly chosen vertex belongs to it. 
    \end{enumerate}
\end{corollary}

With Theorem~\ref{thm: infection occupation} and Corollary \ref{cor: infection occupation}, the percolation process considers occupied components and its size in a probability level to illustrate the condition for epidemic outbreaks to occur.
From the point of view of the new network constructed only by occupied edges, almost all of the results about components from Section~\ref{sec: pgfs} can be applied to occupied components with some modification. 
Since $T$ is uniform for every edge and all occupied component is a subset of components in the original network, this modification can be done as long as we modify all functions as functions also of $T$. 

Suppose that the disease originates from a single randomly chosen vertex.
We want to derive the PGF for the number of occupied edges attached to this vertex, as a function of $T$.
The probability of such a vertex having exactly $m$ occupied edges emerging from it, given its degree $k$, is $\binom{k}{m} T^m (1-T)^{k-m}$ (using the binomial distribution).
Using PGF $G_p$ in \eqref{eqn: pgfgp}, the required PGF is given by,
\begin{align}
\mathcal{G}_p(x;T) & = \sum_{m=0}^{\infty}  \sum_{k=m}^{\infty}p_k \binom{k}{m}  T^m (1-T)^{k-m}x^m  \label{eqn: G0T eqn1} 
    \\
& = \sum_{k=0}^{\infty} p_k \sum_{m=0}^{k} \binom{k}{m}  (xT)^m (1-T)^{k-m}  \label{eqn: G0T eqn1b} 
    \\
& = \sum_{k=0}^{\infty} p_k (1-T+x T)^k \label{eqn: G0T eqn2}
    \\
& = G_p(1+(x-1)T)
        \label{eqn: G0T eqn3}
\end{align}
Observe that if $T=1$, this is the same as the PGFs in \eqref{eqn: pgfgp}. 
This specification works for all the results below, which indicates they are actually generalization of the original network results.
 
Using theorem~\ref{thm: pgf} and \eqref{eqn: pgfGq}, the PGF of the number of occupied edges of a vertex connected to the randomly chosen initial vertex, excluding the confirmed edge, is given by,
\begin{equation}
   \mathcal{G}_q(x;T)=G_q(1+(x-1)T)
    \label{eqn: G1T}
\end{equation}

Let us now consider the outbreak size resulting from infecting a randomly chosen node, which can be interpreted to be the size of occupied component.
If we consider a subgraph only formed by all occupied components, similar threshold for giant and small component in Theorem~\ref{thm: Giant Component Threshold} would also exist for occupied component.
Moreover, Theorem~\ref{thm: component} also applies to the giant occupied component for its uniqueness and existence condition.

With network with large enough size, an outbreak is defined to be an epidemic outbreak if it leads to a significant portion of the population being infected at the end of the outbreak.
This corresponds to the infected component being a giant occupied component. With previous results, we have:
\begin{corollary}
\label{cor: epidemic outbreak}
Under assumption (A1)-(A10), an epidemic outbreak appears if and only if both of the following conditions are satisfied:
    \begin{enumerate}
        \item There exists a unique giant occupied component in the network.

        \item The initial infection happens at one of the occupied vertices in the giant occupied component.
    \end{enumerate}
\end{corollary}

This tells us that the existence of the unique giant occupied component is a necessary condition, but not sufficient.
If (a) is satisfied, then by Corollary~\ref{cor: infection occupation}, the outbreak is epidemic if the giant occupied component is the infected component, with probability equals to the portion of vertices in giant occupied component size among all vertices in the network. 
Moreover, if the epidemic outbreak happened, its final infection size is same as the size of the giant occupied component. 
If there is no giant occupied component, then by Theorem~\ref{thm: component}, all occupied component is small component, thus there is no epidemic regardless of where the infection is initialized.
This is very similar with what we did in Section~\ref{sec: pgfs} for original components, but now we need to modify the results for occupied components.

We again use similar definition of $H_p(x)$ and $H_q(x)$, but focusing on occupied stubs. 
Given a transmissibility of $T$, since every occupied edge is a pair of 2 stubs which can also be labelled as occupied, $T$ is also the prior probability that a randomly chosen stub is occupied. Let $\mathcal{S}_T$ be the portion of vertices belongs to a giant occupied component. 
For applications with limited network size, $\mathcal{S}_T$ is also the ratio of giant occupied component size over the network size. 
Similar to $\mathcal{S}$ in Section~\ref{sec: randomgraphs} as the portion corresponding to giant component, we now consider the case that $\mathcal{S}_T \in [0,1)$, so the giant occupied does not extend to all vertices and there are small occupied components and unoccupied vertices, edges and components exist.
Let $P_s(T)$ be the probability density function of outbreak size $s_T$ when corresponding occupied component is a small component.
Again, we use the idea of conditional probability like \eqref{eqn: hatPs}, let $\hat{P}_s(t)$ as the probability density function of $s_T$ conditioned on it belongs to small component, such that
\begin{equation}
\label{eqn: PsT}
    \hat{P}_s(T)=\frac{1}{1-\mathcal{S}_T}P_s(T)
\end{equation}
Following the definition of $H_p(x)$ in \eqref{eqn: Hp}, we denote the PGF of the conditional small occupied component size by $\mathcal{H}_p(x;T)$, such that
\begin{equation}
    \mathcal{H}_p(x;T)=\sum_{s=0}^\infty \hat{P}_s(T)x^s.
\end{equation}
Analogous to $H_q(x)$ in Section~\ref{sec: pgfs}, we also define $\mathcal{H}_q(x;T)$ to be the PGF for the conditional occupied component size we can reach by following a randomly chosen stubs, given the stubs does not belongs to the giant occupied component.
Analogous to $u$, we define $u_T$ as the the probability that a randomly chosen stub does not attached to a vertex belong to the giant occupied component.
Another interpretation of $u$ is that it is the probability that the vertex attached to the random chosen stub remains uninfected, either in the case that an epidemic outbreak occurred so that the giant occupied component is initially infected or if there is no giant occupied component at all.

We can obtain parallel results to  Theorem~\ref{thm: PGFHq} and Theorem~\ref{thm: PGFHp}, together with the modified PGFs in \eqref{eqn: G0T eqn3} and \eqref{eqn: G1T}, $\mathcal{H}_q(x;T)$ and $\mathcal{H}_p(x;T)$ to derive the following equations:
\begin{align}
\mathcal{H}_p(x;T)&=\frac{x}{1-\mathcal{S}_T} \mathcal{G}_p(u_T \mathcal{H}_q(x;T);T)\label{eqn: H0T}
\\
\mathcal{H}_q(x;T)&=\frac{x}{u_T} \mathcal{G}_q(u_T \mathcal{H}_q(x;T);T)\label{eqn: H1T}
\end{align}
The PGFs $\mathcal{H}_q$ can be found by solving ~\eqref{eqn: H1T}, and $\mathcal{H}_p$ can be found by using the solution of $\mathcal{H}_q$ into \eqref{eqn: H0T}. 
Since $\mathcal{H}_q(x;T)$ and $\mathcal{H}_p(x;T)$ are both PGFs, we can derive the analogous result to Corollary~\ref{cor: ueqn}, and show that $u_T$ should be a solution of,
\begin{equation}
    u_T=\mathcal{G}_q(u_T;T)
    \label{eqn: uT}
\end{equation}
Using Corollary~\ref{cor: Seqn} but changing from $\mathcal{S}$ to  $\mathcal{S}_T$, the probability that a randomly chosen vertex belongs to the giant occupied component, is now given by
\begin{equation}
    \mathcal{S}_T=1-\mathcal{G}_p(u_T;T)
    \label{eqn: UT}
\end{equation}

Once again, $u_T=1$ and $\mathcal{S}_T=0$ together form a trivial solution of \eqref{eqn: uT} and \eqref{eqn: UT}, and of \eqref{eqn: H0T} and \eqref{eqn: H1T}. In this case there is no giant occupied component appears in the network.
If there exists a solution of \eqref{eqn: uT} such that $u_T\in [0,1)$, then there could be a giant occupied component, which means an epidemic does occur.
We can derive an explicit threshold condition based on $T$ that determines whether an epidemic occurs or not.

\begin{theorem} [Adapted from \cite{Newman:2002}]
    If there is no giant component in a large-scale network, the mean component size $\langle s \rangle$ of the network is given by 
    \begin{equation}
        \langle s_T \rangle=1+\frac{\mathcal{G}'_p(1;T)}{1-\mathcal{G}'_q(1;T)}=1+\frac{T G'_p(1)}{1-T G'_q(1)}
    \end{equation}
    \label{thm: sT eqn}
\end{theorem}
\begin{proof}
Suppose that $T$ is such that $u_T={H}_q(1;T)=1$ is the only solution, then there are only small outbreaks with finite size, and no epidemic.
In this case, numerical differentiation of the PGF $\mathcal{H}_p(x;T)$ can be used to determine the specific probability $P_s(T)$ that the outbreak size is $s$. 
Similar contour integration method in \eqref{eqn: Contour Ps} is also suggested as a better option for computing derivatives rather than direct numerical differentiation.
We can always find the mean size $\langle s_T \rangle$ of the small component in the closed form if no giant component exists in the network.
Since $G'_p(1)=\sum_{k=1}^{\infty}k p_k=\langle K \rangle$ and by differential of \eqref{eqn: H0T}, under the situation that $u_T=1, \mathcal{S}_T=0$ we always have:
\begin{equation}
   \langle s_T \rangle=\mathcal{H}'_p(1;T)=\mathcal{G}_p(\mathcal{H}_q(1;T);T)+\mathcal{G}'_p(1;T)\mathcal{H}'_q(x;T)
   \label{eqn: sT1}
\end{equation}
Here we also use the PGF based on normalized degree distribution $p_k$, $\mathcal{G}_p(1;T)=G_p(1)=\mathcal{G}_q(1;T)=G_q(1)=1$, thus we have
\begin{equation}
   \langle s_T \rangle=1+\mathcal{G}'_p(1;T)\mathcal{H}'_q(1;T)
   \label{eqn: sT2}
\end{equation}
Differentiating \eqref{eqn: H1T} and substituting again that $u_T=1, \mathcal{S}_T=0$, we have
\begin{equation}
   \mathcal{H}'_q(1;T)=1+\mathcal{G}'_q(1;T)\mathcal{H}'_q(x;T)  \qquad \Leftrightarrow 
   \qquad
   \mathcal{H}'_q(1;T)=\frac{1}{1-\mathcal{G}'_q(1;T)}
   \label{eqn: H1'T}
\end{equation}
Take \eqref{eqn: H1'T} into \eqref{eqn: sT2} yields the required result.
\end{proof}

Note that $\langle s_T \rangle$ in theorem~\ref{thm: sT eqn} diverges when $T G'_P(1)=1$, this actually provide a transition point about the transmissibility $T_c$.
\begin{theorem} [Adapted from \cite{Newman:2002}]
    The transition threshold $T_c$ of transmissibility $T$ where the finite size occupied component transit into a unique giant occupied component in large-scale network, is given by 
    \begin{equation}
        T_c=\frac{1}{G'_q(1)}=\frac{G'_p(1)}{G''_p(1)}=\frac{\langle K \rangle}{G''_p(1)}
    \end{equation}
    \label{thm: Tc}
\end{theorem}

We will prove that for a network with reasonable complex degree distribution and $T \in (0,1)$, this threshold guarantees the existence and uniqueness of solution of $u_T$ in $(0,1)$ for equation \eqref{eqn: uT}.
For any network with degree distribution satisfied that $\exists k \in \mathbb{Z}^{+} \text{and } k \geq 3$ such that $p_k>0$. Otherwise $p_0+p_1+p_2=1$ give us a graph only with simple structure which we don't always apply percolation method on them. 

Now consider the equation \eqref{eqn: uT}, which equivalent to 
\begin{align}
    f(u_T) & =\mathcal{G}_q(u_T;T)-u_T=G_q(1+(u_T-1)T)-u_T
    \nonumber
    \\
    & =  \frac{1}{\langle K \rangle} G'_p(1+(u_T-1)T)-u_T=0
\end{align}

A solution of \eqref{eqn: uT} is the same as a root of $f(u_T)$.
Let $x=1+(u_T-1)T$.
Since $T \in (0,1)$ and $u_T \in [0,1]$, then $|u_T-1|=1-u_T\in [0,1] \Rightarrow x=1+(u_T-1)T \geq 0$, the equal only happens when $u_T=1$.

Because $G_p(x)= \sum_{k=0}^{\infty}p_k x^k$, then it is a continuous polynomial with non-negative coefficient, so does its finite derivatives.
$p_0 \neq 1 \Rightarrow \langle K \rangle$ is a positive constant.
Then $f(u_T)$ is a polynomial of $u_T$.
Because $1-T>0$ and $G'_p(x)$ has non negative coefficient, $f(0)=\frac{1}{\langle K \rangle} G'_p(1-T) >0$.
Clearly, we also have $f(1)=G_q(1)-1=0$.

Now consider first and second derivative of $f(u_T)$
\begin{align}
    f'(u_T) & =\frac{T}{\langle K \rangle} G''_p(1+(u_T-1)T)-1
    \\
    f''(u_T) & = \frac{T^2}{\langle K \rangle} G'''_p(1+(u_T-1)T)
\end{align}
Thus, $f''(u_T)>0$ for all $u_T \in (0,1)$ since there is at least one positive coefficient in $G'''_p(x)$, which means $f'(u_T)$ increasing strictly in $(0,1)$.

If $T>T_c=\frac{\langle K \rangle}{G''_p(1)}$, then $f'(1)=T \frac{G''_p(1)}{\langle K \rangle}-1>0$.
By $f(1)=0$, $\exists \epsilon \in (0,1)$ s.t. $f(1-\epsilon)<0$.
Since $f$ is continuous and $f(0)>0$, there must be a root of $f(u_T)$ in $(0, 1-\epsilon) \subset (0,1)$ by the first Bolzano-Cauchy theorem.
Uniqueness of the solution is obvious, since $f''(u_T)>0$ indicates that there is at most 1 critical point. 

If $T<T_c$ ,then $f'(1)<0$, $f'(u_T)<0$ for all $u_T \in (0,1)$ since $f''(u_T)>0$. 
Then $f(u,T)$ is strictly decreasing in $[0,1]$, but $f(1)=0$ means there is no solution other than $u_T=1$. 

When $T>T_c$, it is possible that $u_T<1$, which means $\mathcal{S}_T\in (0,1)$ and there is a giant occupied component in the network. 
As discussed in \eqref{eqn: UT}, the proportion of the giant component compared to the whole network is given by:
\begin{theorem} (Adapted from \cite{Newman:2002})
    If there is a giant occupied component in the network, the probability of a randomly chosen vertex belongs to such component with probability:
    \begin{equation}
    \mathcal{S}_T=1-\mathcal{G}_p(u_T;T)
    \end{equation}
\label{thm: ST eqn}
\end{theorem}
In such case, $u_T \in (0,1)$ is the non-trivial solution of equation \eqref{eqn: uT}.

In Corollary~\ref{cor: epidemic outbreak}, $T>Tc$ just indicates (a) the existence of giant occupied component, so the appearance of epidemic outbreak is not guaranteed.
But with Corollary~\ref{cor: infection occupation}, we can conclude that if $T>T_c$, the probability that epidemic outbreak happens is the same as $\mathcal{S}_T$.  
This explains the necessity of (A10) for percolation process: if more than one vertex is initially infectious, then there could be more than one infected component.
Because of the randomness of initial infection, it is hard to conclude the relationship between final infectious number and occupied component size.
Thus we cannot use the existence threshold of giant component to determine the appearance of the epidemic outbreak and approximate the expectation infectious size for small outbreaks.

\section{Dynamics of network models}
\label{sec: dyn}
Newman's typical percolation process discussed in Section~\ref{sec: disease}  focuses on predictions at the end of an epidemic and does not provide any dynamical information about the disease spreading.
It is also limited to configuration network models with assumption (A5) that a network is static so that edges do not change over time.
This is reasonable if the epidemic duration is relatively short so that relationships will remain the same during the interested time period, or the population is a controllable community with highly stable social structure.

In this section we review the modified percolation approach of \cite{MillerSlimVolz:2012} and relate some of the results to the previous sections.
Compared to typical percolation process covered by \cite{Newman:2002}, their modified percolation method~\cite{MillerSlimVolz:2012} is more connected to the original SIR models set up as ODE systems. 
This makes analysis of the models more tractable and provide dynamical information that the typical percolation process lacks.

\cite{MillerSlimVolz:2012} introduced their method on a configuration model with static network first. 
They then allowed the edges change with negligible partnership duration, which leads to Mean-Field Social Heterogeneity (MFSH) model.
These methods are extended to dynamic network models, where edges can change at a constant rate without waiting, and discussed constant changing edges with waiting time. 
They also considered more flexible vertices degree for networks, which assign expected degree to each vertex from a distribution instead of actual degrees.
Edges of these expected degree models are formed with probability proportional to product of 2 expected degree from the 2 vertices on each side of the edge.
This modification allows certain level of randomness for total number of edges and number of edges connect to each vertex. Finally, they also discussed different edge duration models where each time the edges reform with some probability, the degree of related vertices would change accordingly.
These models are more complicated but still mathematical tractable using the modified percolation method.
    
In Section~\ref{sec: CM} we discuss the modified percolation method of~\cite{MillerSlimVolz:2012} applied to CM models with assumption (A5) in order to demonstrate the similarity and difference in its approach compared to typical percolation.
In Section~\ref{sec: MFSH} we review the MFSH model to interpret how modified percolation works with dynamic edges as an example. 
    
\subsection{Configuration Model}\label{sec: CM}
Recall that the configuration model are the static networks that are generated using Algorithm~\ref{alg: Config}.
This is exactly the network used in Newman's typical percolation method that we discuss in the previous sections. Here we use modified percolation and again assume (A1)--(A10). However we will now also rely more on the SIR ODE system which requires more specific assumptions on (A7) and (A8). As previously discussed, our default is (A7a) and (A8a).
We will show that this provides equivalent threshold conditions to the typical percolation process in Section~\ref{sec: perco} under these common assumptions.

Modified percolation uses the same probability generating function $G_p$ for degree distribution of vertices, as defined in \eqref{eqn: pgfgp}.
However, instead of using the average transmissibility $T$, which is averaged over the network, the modified method relates the dynamics on the network to the dynamics of an SIR model set up as an ODE system.
Also, the modified process follows assumption (A6) and treats $S(t)$, $I(t)$ and $R(t)$ as the probability at time $t$ which a random chosen vertex $a$ belongs to the susceptible, infectious and recovered compartments respectively.

In the same way that the transmissibility $T$ in Definition~\ref{def: Transmissibility} relates the transmission process to the network for the typical percolation process, the modified percolation method uses a similar but different probability argument $\theta$ to describe the transmission probability between edges.
But unlike the pre-determined fixed value of $T$, $\theta$ is a function of time $t$ to consider the dynamic of transmission process.

\begin{definition}[Adapted from \cite{MillerSlimVolz:2012}]
\label{def: Theta}
The probability that a randomly chosen partner $b$ of a randomly chosen vertex $a$ in the network has not transmitted to $a$ at time $t$ is given by $\theta(t)$.
Initially, $\theta(0)$ should be close to 1.
\end{definition}

For a reasonably large population size $N$, as in the typical percolation process, it is reasonable to make the simplifying assumption that partners of $a$ are independent.
Therefore, given the degree $K_a=k$ of $a$, this vertex is susceptible at time $t$ with probability 
\begin{equation}
    \mathbb{P}_t(a \in S|K_a=k)=s(k,\theta(t))=\theta(t)^k
\label{eqn: CM prob s condition on k}
\end{equation}
Using the law of total probability, we have $S(t)$ defined by $\theta$ from the network so that,
\begin{equation}
    S(t)=\mathbb{P}_t(a \in S)=\sum_{k=0}^{\infty} \mathbb{P}_t(a \in S|K_a=k) \mathbb{P}(K_a=k)=\sum_{k=0}^{\infty} p_k \theta(t)^k=G_p(\theta(t)).
\label{eqn: CM S(t)}
\end{equation}
Thus we derive the following system of equations based on ODE system of MA-SIR model,
\begin{equation}
\label{eqn: CM ODE}
    \begin{cases}
        S(t) & = G_p(\theta(t)) \\
        I(t) & = 1-S(t)-R(t) \\
        \dot{R}(t) & = \hat{\gamma} I(t) 
    \end{cases} 
\end{equation}
The problem now is to describe dynamic of $\theta(t)$.    
In general, $\theta(t)$ depends on the degree of $a$, changing rate of partners and probability that a random partner $b$ of vertex $a$ is infected. 
As discussed in Section~\ref{sec: pgfs}, because a random partner $b$ of $a$ has a confirmed edge, $b$ is likely to have more partners than any randomly chosen vertex. 
So the infected fraction $I(t)$ does not provide the infection probability of $b$.
    
In the configuration model, to describe $\theta(t)$, we separate a general neighbour vertex $b$ into four different classes with corresponding proportion of all partners, which also equivalent to the probability that a random $b$ belongs to each class.
By definition of $\theta$ partners that have already been infected and have already transmitted to $a$ will have probability $1-\theta$.
The rest of the probability $\theta$ now can be broken down into three classes with probability $\Phi_S$, $\Phi_I$ and $\Phi_R$, which represent the probability that a random partner of vertex $a$ belongs to each corresponding compartment ($S$, $I$ and $R$), and has not transmitted infection to $a$. The sum of $1-\theta$, $\Phi_S$, $\Phi_I$ and $\Phi_R$ must equal one yielding the relationship,
\begin{equation}
    \theta=\Phi_S+\Phi_I+\Phi_R.
\label{eqn: CM theta breakdown}
\end{equation}

Since the edges do not change with time in configuration models, the rate of change of the compartments only depends on the flow of neighbourhood vertices between compartments. 
This system of equations is illustrated in Figure~\ref{fig:CM}.
The flow from $\Phi_I$ to $1-\theta$ is $\hat{\beta}\Phi_I$, where the the per-infected transmission rate $\hat \beta$ is the rate an infected partner transmits to $a$. 
Then,
\begin{equation}
    \dot{\theta}(t)=-\hat{\beta} \Phi_I
\label{eqn: CM theta PhiI ODE}
\end{equation}
To find $\Phi_I$, we use \eqref{eqn: CM theta breakdown} to get $\Phi_I=\theta-\Phi_S-\Phi_R$. We then compute $\Phi_S$, and $\Phi_R$ explicitly.
First, we note that the flow from $\Phi_I$ to $\Phi_R$ is $\hat{\gamma} \Phi_I$, where $\hat{\gamma}$ is the rate of an infected partner recovers.
Therefore this rate is proportional to the rate going into $1-\theta$ with the constant factor $\hat{\gamma} /\hat{\beta}$.
As $\Phi_R$ and $1-\theta$ both have approximately zero initial values then,
\begin{equation}
    \Phi_R=\frac{\hat{\gamma}}{\hat{\beta}}(1-\theta)
\label{eqn: CM PhiR eqn}
\end{equation}

\begin{figure}[htbp]
    \begin{center}
        \begin{tikzpicture}[->, >=stealth',shorten >=1pt,auto, node distance=3cm, thick, main node/.style={rectangle,draw,minimum size=1cm, font = \sffamily\bfseries}]

            \node[main node] (S) {$\Phi_S=\frac{G'_p(\theta)}{G'_p(1)}$};
            \node[main node] (I) [right of=S] {$\Phi_I$};
            \node[main node] (R) [right of=I] {$\Phi_R$};
            \node[main node] (Q) [below =1cm of I] {$1-\theta$};
  
            \draw[thick,->] (S) -- (I) node[midway,sloped,above,rotate=0]{};
            \draw[thick,->] (I) -- (R) node[midway,sloped,above,rotate=0]{$\hat{\gamma}\Phi_I$};
            \draw[thick,->] (I.south) -- (Q.north) node[midway,sloped,right,rotate=90] {$\hat{\beta}\Phi_I$};
	
        \end{tikzpicture}
        \caption[Flow for the configuration model]{Flow for the configuration model. 
        Figure was created using Latex Tikz Picture}
        \label{fig:CM}
    \end{center}
\end{figure}
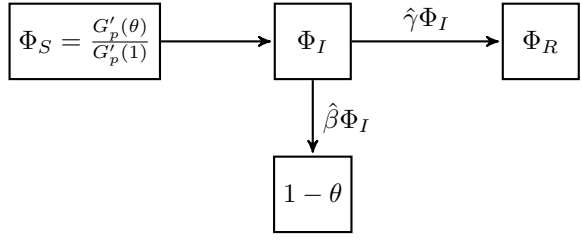

Now consider the probability that a random partner $b$ of the chosen vertex $a$ has degree $K_b=k$. This probability is equal to the total number of stubs belongs to all degree $k$ vertices divide by number of all stubs. We use the notation,
\begin{equation}
    Q_b(k)=\mathbb{P}(K_b=k)=\frac{N p_k k}{\sum_{n=0}^{\infty} N p_n n}=\frac{k p_k}{\langle K \rangle}
\label{eqn: CM G1pdf}
\end{equation}
where $p_k$ is the probability mass function of the vertex degree distribution.
Given $k$, the random partner $b$ is susceptible with probability $\mathbb{P}(b \in S|K_b=k)=\theta^{k-1}$, since $a$ cannot transmit to $b$.
Thus, by law of total probability, we have
\begin{equation}
    \Phi_S=\sum_{k=0}^{\infty} \mathbb{P}(K_b=k) \mathbb{P}(b \in S|K_b=k)=\sum_{k=0}^{\infty}\frac{k p_k \theta^{k-1}}{\langle K \rangle}=\frac{G'_p(\theta)}{G'_p(1)}
\label{eqn: CM PhiS eqn}
\end{equation}

Equations \eqref{eqn: CM theta breakdown}, \eqref{eqn: CM PhiR eqn} and \eqref{eqn: CM PhiS eqn} combined with ODE \eqref{eqn: CM theta PhiI ODE} yields a closed form ODE system for $\theta(t)$ given by,
\begin{equation}
    \dot{\theta}(t)=-\hat{\beta} \theta +\hat{\beta} \frac{G'_p(\theta)}{G'_p(1)}+\hat{\gamma} (1-\theta).
    \label{eqn: CM theta ODE}
\end{equation}
    
By solving \eqref{eqn: CM theta ODE} together with \eqref{eqn: CM ODE}, we can get similar results as we found in Section~\ref{sec: perco}.
An example of the trajectory of the infectious compartment $I(t)$ is shown in Figure~\ref{fig: epidemic_curves}.

Like the epidemic threshold $T_c$ in Theorem~\ref{thm: Tc} of transmissibility $T$, we can find a similar threshold for epidemics, but using the basic reproductive number as in the MA-SIR model.

\begin{definition}[\cite{MillerSlimVolz:2012}]
    The \textbf{basic reproduction number} $\mathcal{R}_0$ is the expected number of infections from a single infected vertex.
\end{definition}
When $\mathcal{R}_0 < 1$, epidemic outbreaks (leading to a giant component) are impossible and the method breaks down. If $\mathcal{R}_0 > 1$, epidemic outbreaks are possible but not guaranteed.
We have the following result for $\mathcal{R}_0$
\begin{theorem}[Adapted from \cite{MillerSlimVolz:2012}]
\label{thm: R0 CM}
    For the configuration model under assumption (A1)-(A10) together with (A7a) and (A8a), the basic reproductive number is given by
    \begin{equation}
        \mathcal{R}_0=\sum_{k=0}^{\infty} Q_b(k) (k-1) \frac{\hat{\beta}}{\hat{\beta}+\hat{\gamma}}=\frac{\hat{\beta}}{\hat{\beta}+\hat{\gamma}} \frac{\langle K^2-K \rangle}{\langle K \rangle}=\frac{\hat{\beta}}{\hat{\beta}+\hat{\gamma}}\frac{G''_p(1)}{G'_p(1)}
        \label{eqn: CM R0}
    \end{equation}
    where $\frac{\hat{\beta}}{\hat{\beta}+\hat{\gamma}}$ is the probability a vertex infects a neighbor prior to recovering.
\end{theorem}

The $(k-1)$ term exist since newly infected vertex must have a confirmed partner who is not susceptible as its infection source.

Compare this to the transmissibility $T=\frac{\hat{\beta}}{\hat{\beta}+\hat{\gamma}}$ in Corollary~\ref{cor: const T} with same assumption, we can easily show the equivalence of the threshold property in both cases in the following corollary:

\begin{corollary}
\label{cor: threshold equiv}
Under same assumptions (A1)-(A10), (A7a) and (A8a), the threshold of $\mathcal{R}_0=1$ given by Theorem~\ref{thm: R0 CM} for the modified percolation model is equivalent to the threshold $T=T_c$ given by Theorem~\ref{thm: Tc}.
\end{corollary}
\begin{proof}
    Since $T=\frac{\hat{\beta}}{\hat{\beta}+\hat{\gamma}}$, and $T_c=\frac{G'_p(1)}{G''_p(1)}$, we have
    \begin{equation}
    \label{eqn: threshold equiv}
    \mathcal{R}_0=1 =\frac{\hat{\beta}}{\hat{\beta}+\hat{\gamma}}\frac{G''_p(1)}{G'_p(1)}= T \times \frac{1}{T_c}
    \quad \Leftrightarrow \quad T=T_c
    \end{equation}
\end{proof}

To calculate the final size of the epidemic, we solve for the equilibria of the ODE system \eqref{eqn: CM ODE}.
If $\mathcal{R}_0 > 1$ and the epidemic does happen, there are two possible equilibrium solutions of \eqref{eqn: CM theta ODE}.
One of them will be $\theta = 1$ which occurs when the disease has not been introduced.
At another condition where $\theta<1$, the disease has spread first and then dies out leading to a a nontrivial trajectory of the epidemic with $\theta$ going to,
\begin{equation}
    \theta(\infty)=\frac{\hat{\gamma}}{\hat{\beta}+\hat{\gamma}}+\frac{\hat{\beta}}{\hat{\beta}+\hat{\gamma}}\frac{G'_p(\theta(\infty))}{G'_p(1)}.
\label{eqn: CM theta infty}
\end{equation}
    
\begin{theorem}[Adapted from \cite{MillerSlimVolz:2012}]
\label{thm: Rinfty CM}
     For the configuration model under assumption (A1)-(A10) together with (A7a) and (A8a),the total fraction of the population infected at the end of epidemic is,
    \begin{equation}
        \mathcal{R}(\infty)=1-G_p(\theta(\infty)).
        \label{eqn: CM final size}
    \end{equation}
\end{theorem}

\cite{MillerSlimVolz:2012} stated that $\mathcal{R}(\infty)$ in Theorem~\ref{thm: Rinfty CM} is in agreement with $\mathcal{S}_T$ in Theorem~\ref{thm: ST eqn}.
It is evident that they are similar in structure, and here we provide the details of the proof.

\begin{theorem}[Connecting Theorem~\ref{thm: ST eqn} and Theorem~\ref{thm: Rinfty CM}]
\label{thm:equiv}
For the CM model under assumption (A1)-(A10) together with (A7a) and (A8a), we have,
\begin{enumerate}
\item $\theta(\infty)=1+(u_T-1)T$
\item $\mathcal{R}(\infty)=\mathcal{S}_T$
\end{enumerate}
\end{theorem}
\begin{proof} To prove part 1 we first recall the expression we need to solve for $u_T$ in \eqref{eqn: uT},
\[
u_T=\mathcal{G}_q(u_T;T)=G_q(1+(u_T-1)T)
=\frac{G_p'(1+(u_T-1)T)}{G_p'(1)}.\]
Thus,
\[
1+(u_T-1)T=1+\Big[\frac{G_p'(1+(u_T-1)T)}{G_p'(1)}-1\Big]T   .
\]
From Corollary~\ref{cor: const T} we have $T=\frac{\hat\beta}{\hat \beta +\hat \gamma}$ which leads to,
\begin{align*}
1+(u_T-1)T
&=1+\Big[\frac{G_p'(1+(u_T-1)T)}{G_p'(1)}-1\Big]\frac{\hat\beta}{\hat \beta +\hat \gamma}\\
1+(u_T-1)T&=\frac{\hat\gamma}{\hat \beta +\hat \gamma}+\frac{G_p'(1+(u_T-1)T)}{G_p'(1)}\frac{\hat\beta}{\hat \beta +\hat \gamma}\\
\end{align*}
This expression for $1+(u_T-1)T$ is exactly the same expression in \eqref{eqn: CM theta infty} that $\theta(\infty)$ has to satisfy.
This completes the proof of part 1.
Part 2 easily follows from part 1, Theorem~\ref{thm: ST eqn} and \eqref{eqn: CM final size}.
\end{proof}

In Figure~\ref{fig: MSV vs NSW}, we also plot $\mathcal{R}(\infty)$ and $\mathcal{S}_T$ numerically under same given parameters, for the following commonly used degree distributions: (a) discrete Poisson $p_k=e^{-\lambda} \lambda^k/k!$, (b) discrete exponential $p_k=(1-e^{-\alpha}) e^{-\alpha (k-1)}$ and (c) truncated power-law $p_k=k^{-\gamma}/\zeta(\gamma)$.
\begin{figure}[htbp]
    \begin{center}
        \includegraphics[width=1\textwidth]{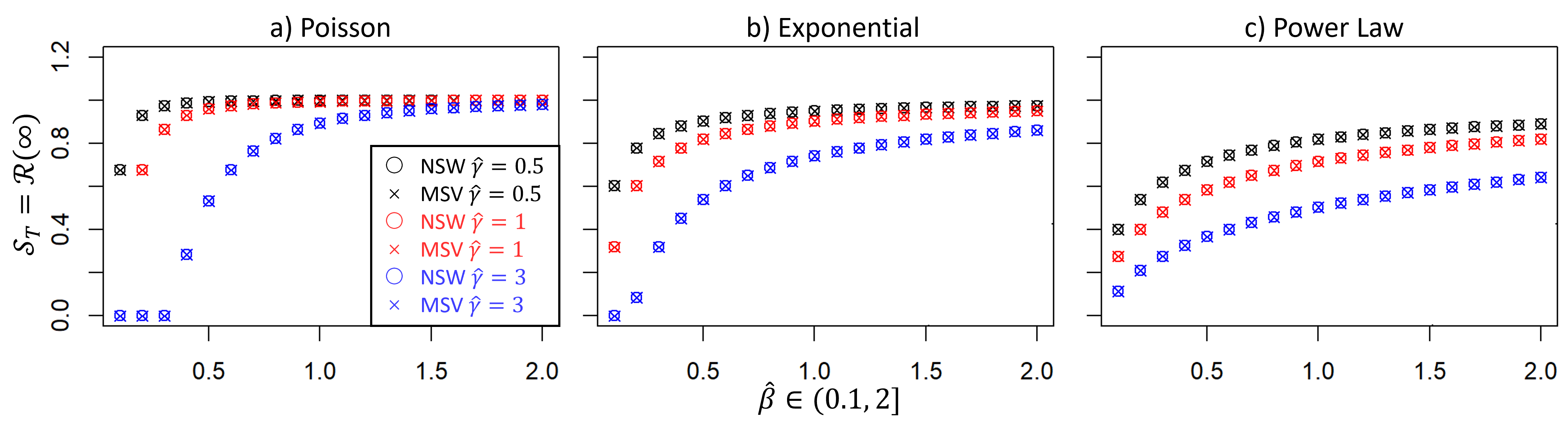}
    \end{center}
    \caption[Comparison of MSV and NSW]{Comparison of $\mathcal{S}_T$ and $\mathcal{R}(\infty)$: numerical result of typical percolation models and modified percolation models for a) discrete Poisson b) discrete exponential and c) truncated power-law distribution, all with mean degree $\langle K \rangle =10$ and bounded maximum degree $\Delta \leq 200$.
    The $X$-axis is the value of $\hat{\beta}$, where different color represent different fixed value for $\hat{\gamma}$: black is $\hat{\gamma}=0.5$, red is $\hat{\gamma}=1$ and blue is $\hat{\gamma}=3$.
    Circles denote results  for the \cite{NewmanStrogatzWatts:2001} model (NSW) and crosses are results using the \cite{MillerSlimVolz:2012} model (MSV). The agreement in the results show their equivalence, as proven in Theorem~\ref{thm:equiv}.
    Figures were created using our R packages.
    }
    \label{fig: MSV vs NSW}
\end{figure}

Compared to MA-SIR models, network-SIR models can capture substantially more population structure. 
The edge-based compartmental modeling approach allows us to do this with just marginal more complexity since we end up with an ODE system.
Compared to the typical percolation models, modified percolation also provides more dynamic information, like the epidemic trajectory. 

\subsection{Mean Field Social Heterogeneity Model}\label{sec: MFSH}
In this section we briefly  mention how the modified percolation process can further extend application of edge-based compartmental modeling to non-static dynamic networks in which the structure of network change with time.
As an example, we present the MFSH network model presented by \cite{MillerSlimVolz:2012}.
In such a model we introduce simple dynamics based on configuration networks, and alter the assumption (A5) with negligible partnership duration as showed in Figure~\ref{fig:MFSH networks}.

\begin{figure}[htbp]
    \begin{center}
        \includegraphics[width=0.7\textwidth]{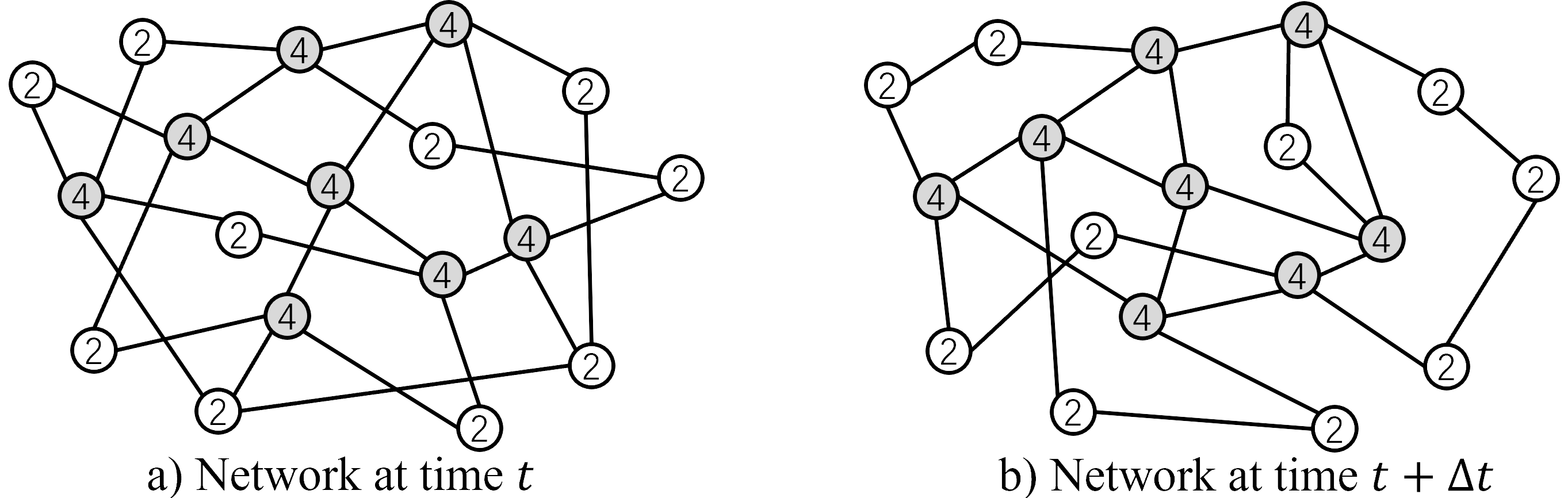}
    \end{center}
    \caption[Example of MFSH networks]{Example of MFSH networks: Given degree sequence such that half vertices(grey) have degree 4 and half vertices(white) have degree 2. 
    a) is network at random time $t$. b) is network at $t+\Delta t$ with dynamics on network structure, after an arbitrarily small time period $\Delta t$. 
    Degree of each vertex would not change with time, but edges always immediately break into stubs and reform new edges at any time.
    Figures were created using MS Office}
    \label{fig:MFSH networks}
\end{figure}

For MFSH networks, the degree assignment, compartment flow and corresponding PGF will be exactly the same as in CM networks. 
However, the stubs assigned to each vertex are reformed independently into edges at every moment.
Since the edges are no longer static, the definition of $\theta(t)$ is slightly modified but with similar approach as in the previous section and some modification, one can obtain a modified SIR system with respect to $\theta$ that has the exact same form with \eqref{eqn: CM ODE}.
The full details of this system are given in \cite{MillerSlimVolz:2012}.
This is a very interesting direction for further study and applications.

\section{Summary and Future Work}
\label{sec:summary}
In this paper we reviewed network models for epidemiology, connected them to ideas from bond percolation process, and filled in some of the gaps in their relationship with their theoretical basis in random graph theory. The network models we reviewed were the typical bond percolation model (NSW) created by \cite{NewmanStrogatzWatts:2001} and more widely-used modified bond percolation model (MSV) presented by \cite{MillerSlimVolz:2012}.

In Section \ref{sec: basic} we began by reviewing the definitions and theorems about components and thresholds in random graph theory.
Their original description in the pioneering network epidemiology papers were not detailed and this led to some ambiguity in the literature and confusion for researchers looking to apply the given models. 
Here we presented a systematic discussion about the network generation methods modified from the well-known configuration model, and considered the realization condition for the degree sequences.
Moreover, in Section~\ref{sec: pgfs} we present a detailed discussion about the equivalence result of giant component size and thresholds using alternative PGFs method, which is the core methodology of the bond percolation model.

In Section~\ref{sec: disease} we provided a list of fundamental assumptions for the edge-based random network models for epidemiology.
This helped clarify issues between the different assumptions made in network epidemiology.
In Section~\ref{sec: Trans}, we also illustrated how random graph theory is related to disease transmission, using the assumptions and percolation process with transmissibility $T$.
We also discussed the impact of assumptions to the model by comparing a series of different $T$ expression under different assumptions, this also provide an example framework for modification of the network model to fit vary cases depends on complexity in reality.
Later in Section~\ref{sec: perco}, we derived an equivalent result to \cite{Newman:2002}, but now on a more rigorous basis.
We also provided a simple proof about the uniqueness of the solution corresponding to epidemic outbreak and showed its existence condition is equivalent of the epidemic threshold generate by the model. 

In Section~\ref{sec: dyn}, we briefly discussed the modified percolation MSV model, which  follows similar edge-based and PGF techniques, but in addition provides dynamical information about the disease transmission using ODE system.
We start with presenting the MSV model result under same static network assumptions of the NSW model in Section~\ref{sec: CM}, and show the equivalence of the two model by showing the conditions for which they have the same thresholds and final epidemic sizes.
Then in Section~\ref{sec: MFSH}, we briefly introduce the MFSH network model to show the potential of the modified MSV model on non-static networks.

As discussed in both Section~\ref{sec: basic} and Section~\ref{sec: disease}, The random network models use mean-field ideas, which approximate the stochastic network construction and disease transmission by deterministic result of expectations of all possible results.
A natural question is the exactness of the mean field approximation in these models, which often represented by the large $N$ limit of the model prediction: the stochastic result should converge w.h.p. to the deterministic expectation result.

Both the NSW model in \cite{NewmanStrogatzWatts:2001} and the MSV model in \cite{MillerSlimVolz:2012} provide a good numerical performance in exactness from stochastic simulation of large network.
\cite{Decreusefond:2012} has proved that MSV model satisfy the large $N$ limit for configuration models, but analysis proof for NSW model requires further investigation.
But, by Corollary~\ref{cor: threshold equiv} and Figure~\ref{fig: MSV vs NSW}, we partially proved the equivalence of the two models in CM networks, thus the exactness of the result from MSV model can be converted to typical model indirectly.
Moreover, in Section~\ref{sec: pgfs}, we show that the result in Corollary~\ref{cor: Seqn} provide by PGFs method in NSW model is equivalent to the result in Theorem~\ref{thm: Giant Component Size} by \cite{MolloyReed:1995}, which provide the exactness of mean-field result for stochastic in network construction.

Later results by \cite{KissKenahRempala:2023} also proved the equivalence of MSV model for CM network with other random network approaches, like vertex pairwise closure model by \cite{Keeling:1999} under specific type of degree distribution, and more recently dynamical survival analysis model by \cite{Jacobsen:2016}.
These approaches are generated separately from different mathematical theorems and ideas, and each has their own advantages in different applications of network models.
The equivalence reveals both the hidden mathematical relations in different fields and the huge potential of edge-based models in application.
Our work further linked these models to earlier basic models and the theoretical foundations of graph theory.
Therefore, we hope it could help researchers to understand and use the models in both theoretical and application level.

Here we also present a package EpiNetPerco in R (https://github.com/RichardSichengZhao/R-EpiNetPerco), which automates the calculation and simulation of bond-percolation network models for epidemiology.
The package contains the prediction function for NSW and MSV models, network generation realizations based on configuration model algorithm and some sequential generation algorithms inspired by \cite{BlitzsteinDiaconis:2011} and \cite{BayatiMohsenKim:2010}, and the stochastic transmission simulator on network based on Doob–Gillespie algorithm (also known as Stochastic Simulation Algorithm, SSA).
We are also looking to improve the performance of the package for large scale network generation and simulation, by optimizing data structure and parallel computing, so that users can apply edge-based epidemic modeling to their research more easily. 

\section*{Statements and Declarations}

\subsubsection*{Data Availability} \quad The simulated datasets and code used for the study will be available on Github with the package EpiNetPerco.

\subsubsection*{Conflict of Interests} \quad The authors have no relevant financial or non-financial interests to disclose.

\begin{appendices}
\section{Probability generating functions}
\label{sec: pgf}
Suppose that a random variable $X$ has probability mass function given by $p_k=P(X=k)$ for $k\in\mathbb{Z}_{\geq 0}$. This random variable has probability generating function given by,
\begin{equation}
\label{eqn:pgfdef}
    G_X(x)=E(x^X)=\sum_{k=0}^{\infty}p_k x^k.
\end{equation}

\begin{theorem}\label{thm: pgf}
Suppose that $X$ has PGF given by $G_X(x)$, and that $X_i$ are independent and identically distributed random variables with the same distribution as $X$. 
The following statements hold:
\begin{enumerate}
    \item The moments of $X$ is given by $E (X^n )= \sum_{k=0}^{\infty} p_k X^n = \lim_{x\to 1^-}\Big[(x \frac{d}{dx})^n G(x)\Big]$.
    \item The PGF of the random variable $X+1$ is $G_{X+1} (x)=x G_x(x)$.
    \item If $N \in \mathbb{Z}_{>0}$ is fixed then the PGF of $S_N=\sum_{n=1}^N X_i$ is given by $G_{S_N}=(G_X(x))^N$.
\item If $N$ is a discrete random variable with probability generating function $G_N(x)$ then the PGF of $S_N=\sum_{n=1}^N X_i$ is given by $G_{S_N}=G_N(G_X(x))$.
\end{enumerate}
\end{theorem}

\begin{proof}
These are standard properties of PGFs. See \cite{JohnsonKempKotz:2005}
\end{proof}
\end{appendices}

\newpage
\bibliography{Bibliography.bib}

\end{document}